\documentclass{theoretics}
\usepackage[T1]{fontenc}
\usepackage{comment}
\usepackage[english]{babel}
\usepackage{enumitem}
\usepackage{mathrsfs}
\usepackage{todonotes}

\definecolor{MyBlue}{RGB}{50,100,200}
\usepackage[capitalise,noabbrev,nameinlink]{cleveref}

\crefformat{page}{#2page~#1#3}%
\Crefformat{page}{#2Page~#1#3}%
\crefformat{equation}{#2(#1)#3}%
\Crefformat{equation}{#2(#1)#3}%
\crefformat{figure}{#2Figure~#1#3}%
\Crefformat{figure}{#2Figure~#1#3}%
\crefformat{section}{#2Section~#1#3}
\Crefformat{section}{#2Section~#1#3}
\crefformat{chapter}{#2Chapter~#1#3}
\Crefformat{chapter}{#2Chapter~#1#3}
\crefformat{chapter*}{#2Chapter~#1#3}
\Crefformat{chapter*}{#2Chapter~#1#3}
\crefformat{part}{#2Part~#1#3}
\Crefformat{part}{#2Part~#1#3}
\crefformat{enumi}{#2(#1)#3}
\Crefformat{enumi}{#2(#1)#3}

\crefname{techclaim}{Claim}{Claims}
\crefname{observation}{Observation}{Observations}

\usepackage{thm-restate}

\renewcommand{\leq}{\leqslant}
\renewcommand{\geq}{\geqslant}

\renewcommand{\phi}{\varphi}

\newcommand{\cfleq}{\sqsubseteq}

\newcommand{\Bb}{\mathcal{B}}
\newcommand{\Dd}{\mathcal{D}}

\newcommand{\Z}{\mathbb{Z}}
\newcommand{\Zn}{\Z_{\geq 0}}
\newcommand{\R}{\mathbb{R}}
\newcommand{\Q}{\mathbb{Q}}
\newcommand{\myRn}{\R_{\geq 0}} 

\newcommand{\tup}[1]{{\mathbf{#1}}}
\newcommand{\xx}{\tup x}
\newcommand{\ww}{\tup w}
\newcommand{\whww}{{\widehat{\ww}}}
\newcommand{\yy}{\tup y}
\newcommand{\zz}{\tup z}
\newcommand{\sv}{\tup s}
\newcommand{\tv}{\tup t}
\newcommand{\av}{\tup a}
\newcommand{\bb}{\tup b}
\newcommand{\ff}{\tup f}
\renewcommand{\gg}{\tup g}
\newcommand{\hh}{\tup h}
\newcommand{\cc}{\tup c}
\newcommand{\rr}{\tup r}
\newcommand{\dd}{\tup d}
\newcommand{\ee}{\tup e}
\newcommand{\pp}{\tup p}
\newcommand{\qq}{\tup q}
\newcommand{\uu}{\tup u}
\newcommand{\myvv}{\tup v} 
\newcommand{\matrices}[3]{#1^{#2\times #3}}

\newcommand{\Graver}{\mathsf{Graver}}

\newcommand{\zero}{\tup 0}

\newcommand{\mset}[1]{\left\langle #1\right\rangle}

\newcommand{\Oh}{{\mathcal{O}}}

\newcommand{\Qq}{\mathcal{Q}}
\newcommand{\Ff}{\mathcal{F}}
\newcommand{\Gg}{\mathcal{G}}

\newcommand{\Rr}{\mathcal{R}}
\newcommand{\Zz}{\mathcal{Z}}
\newcommand{\Ss}{\mathscr{L}}

\newcommand{\RHS}{\mathsf{RHS}}
\newcommand{\Diag}{\mathsf{Diag}}

\newcommand{\Base}{\mathsf{Base}}
\newcommand{\Gra}{\mathsf{Gra}^+}
\newcommand{\rev}{\mathsf{best}}
\newcommand{\Types}{\mathsf{Types}}

\newcommand{\cntTypes}{\mathsf{count}}

\newcommand{\Tight}{\mathsf{Tight}}

\newcommand{\sca}[2]{\langle #1,#2\rangle}

\newcommand{\nfoldFeas}{\textsc{$n$-fold ILP Feasability}}

\newcommand{\nfoldOpt}{\textsc{$n$-fold ILP Optimization}}

\newcommand{\stageFeas}{\textsc{Two-Stage Stochastic ILP Feasability}}

\newcommand{\nfold}{$n$-fold}

\newcommand{\stage}{two-stage stochastic}

\renewcommand{\setminus}{-}

\newcommand{\realcone}{\mathsf{cone}}
\newcommand{\intcone}{\mathsf{intCone}}
\newcommand{\lattice}{\mathsf{lattice}}
\newcommand{\spn}{\mathsf{span}}

\makeatletter
\newcommand\IfRestateTF{%
  \ifx\label\thmt@gobble@label 
    \expandafter\@firstoftwo
  \else
    \expandafter\@secondoftwo
  \fi
}
\makeatother
\newcommand{\RestateRemark}{\IfRestateTF{{\normalfont\bfseries (Restated) }}{}}

\title{Parameterized algorithms for block-structured integer programs with large entries}

\ThCSauthor[epfl]{Jana Cslovjecsek}{cslovj@gmail.com}[]
\ThCSauthor[prague]{Martin Kouteck\'y}{koutecky@iuuk.mff.cuni.cz}[0000-0002-7846-0053]
\ThCSauthor[eindhoven]{Alexandra Lassota}{a.a.lassota@tue.nl}[0000-0001-6215-066X]
\ThCSauthor[warsaw]{Micha\l{}~Pilipczuk}{michal.pilipczuk@mimuw.edu.pl}[000-0001-7891-1988]
\ThCSauthor[milan]{Adam Polak}{adam.polak@unibocconi.it}[0000-0003-4925-774X]

\ThCSaffil[epfl]{École Polytechnique Fédérale de Lausanne (EPFL), Switzerland}
\ThCSaffil[prague]{Computer Science Institute, Charles University, Prague, Czech Republic}
\ThCSaffil[eindhoven]{Eindhoven University of Technology, Eindhoven, The Netherlands}
\ThCSaffil[warsaw]{Institute of Informatics, University of Warsaw, Poland}
\ThCSaffil[milan]{Department of Computing Sciences, Bocconi University, Milan, Italy}

\ThCSthanks{A preliminary version of this article appeared at SODA 2024 \cite{CslovjecsekKLPP24}. Martin Koutecký was partially supported
by Charles University project UNCE/SCI/004 and by the project 22-22997S of GA ČR. The work of Micha\l{} Pilipczuk on this manuscript is a part of a project that has received funding from the European Research Council (ERC) under the European Union's Horizon 2020 research and innovation programme, Grant Agreement No 948057 --- BOBR. Part of this work was done when Adam Polak and Alexandra Lassota were at École Polytechnique Fédérale de Lausanne, supported by the Swiss National Science Foundation projects \emph{Lattice Algorithms and Integer Programming} (185030) and \emph{Complexity of Integer Programming} (CRFS-2\_207365). Part of this work was carried out while Alexandra Lassota was affiliated with Max Planck Institute for Informatics, SIC, Saarbrücken, Germany.}

\ThCSshortnames{J.\ Cslovjecsek, M.\ Kouteck\'y, A.\ Lassota, M.\ Pilipczuk, A.\ Polak}  
\ThCSshorttitle{Parameterized algorithms for block-structured integer programs with large entries}
\ThCSyear{2025}
\ThCSarticlenum{15}
\ThCSreceived{Jan 21, 2024}
\ThCSrevised{Dec 3, 2024}
\ThCSaccepted{Feb 4, 2025}
\ThCSpublished{Jul 18, 2025}
\ThCSdoicreatedtrue
\ThCSkeywords{block-structured integer programming, parameterized complexity, Graver basis}

\begin{document}
\maketitle


\begin{abstract}
We study two classic variants of block-structured integer programming. {\em{Two-stage stochastic programs}} are integer programs of the form $\{A_i \xx + D_i \yy_i = \bb_i\textrm{ for all }i=1,\ldots,n\}$, where $A_i$ and $D_i$ are bounded-size matrices. Intuitively, this form corresponds to the setting when after setting a small set of global variables $\xx$, the program can be decomposed into a possibly large number of bounded-size subprograms.
On the other hand, {\em{$n$-fold programs}} are integer programs of the form $\{{\sum_{i=1}^n C_i\yy_i=\av} \textrm{ and } D_i\yy_i=\bb_i\textrm{ for all }i=1,\ldots,n\}$, where again $C_i$ and $D_i$ are bounded-size matrices. This form is natural for knapsack-like problems, where we have a large number of variables partitioned into small-size groups, each group needs to obey some set of local constraints, and there are only a few global constraints that link together all the variables.

A line of recent work established that the optimization problem for both two-stage stochastic programs and $n$-fold programs is fixed-parameter tractable when parameterized by the dimensions of relevant matrices $A_i,C_i,D_i$ and by the maximum absolute value of any entry appearing in the constraint matrix. A fundamental tool used in these advances is the notion of the {\em{Graver basis}} of a matrix, and this tool heavily relies on the assumption that all the entries of the constraint matrix are bounded.

In this work, we prove that the parameterized tractability results for two-stage stochastic and $n$-fold programs persist even when one allows large entries in the global part of the program. More precisely, we prove the~following:
\begin{itemize}[nosep]
 \item The feasibility problem for two-stage stochastic programs is fixed-parameter tractable when parameterized by the dimensions of matrices $A_i,D_i$ and by the maximum absolute value of the entries of matrices $D_i$. That is, we allow matrices $A_i$ to have arbitrarily large entries.
 \item The linear optimization problem for $n$-fold integer programs that are uniform -- all matrices $C_i$ are equal -- is fixed-parameter tractable when parameterized by the dimensions of matrices $C_i$ and $D_i$ and by the maximum absolute value of the entries of matrices $D_i$. That is, we require that $C_i=C$ for all $i=1,\ldots,n$, but we allow $C$ to have arbitrarily large entries.
\end{itemize}
In the second result, the uniformity assumption is necessary; otherwise the problem is $\mathsf{NP}$-hard already when the parameters take constant values. Both our algorithms are weakly polynomial: the running time is measured in the total bitsize of the input.

In both results, we depart from the approach that relies purely on Graver bases. Instead, for two-stage stochastic programs, we devise a reduction to integer programming with a bounded number of variables using new insights about the combinatorics of integer cones.
For $n$-fold programs, we reduce a given $n$-fold program to an exponential-size program with bounded right-hand sides, which we consequently solve using a reduction to mixed integer programming with a bounded number of integral variables.
For $n$-fold programs, we reduce a given $n$-fold program to an exponential-size program with bounded right-hand sides, which we consequently solve using a reduction to mixed integer programming with a bounded number of integral variables.
\end{abstract}

\section{Introduction}\label{sec:intro}

We study two variants of integer programming problems, where the specific structure of the constraint matrix can be exploited for the design of efficient parameterized algorithms. {\em{Two-stage stochastic programs}} are integer programs of the following form:
\begin{align}
 \xx\in \Zn^k,\ \yy_i\in \Zn^k, &\qquad  \textrm{and}\nonumber\\
 A_i\xx + D_i\yy_i =\bb_i &\qquad \textrm{for all }i=1,2,\ldots,n.\tag{$\spadesuit$}\label{eq:twostage-form}
\end{align}
Here, $A_i,D_i$ are integer $k\times k$ matrices\footnote{Reliance on square matrices is just for convenience of presentation. The setting where blocks are rectangular matrices with dimensions bounded by $k$ can be reduced to the setting of $k\times k$ square matrices by just padding with $0$s.} and each $\bb_i$ is an integer vector of length $k$. This form arises naturally when the given integer program can be decomposed into multiple independent subprograms on disjoint variable sets $\yy_i$, except there are several global variables $\xx$ that are involved in all the subprograms and thus link them. See the survey of Shultz et al.~\cite{SchultzSvdV96} as well as an exposition article by Gaven\v{c}iak et al.~\cite{GavenciakKK22} for examples of~applications.

We moreover study {\em{$n$-fold programs}} which are integer programs of the form
\begin{align}
    \yy_i\in \Zn^k, &\nonumber \\
    \sum_{i=1}^n C_i\yy_i = \av,\nonumber & \qquad \textrm{and} \tag{$\clubsuit$}\label{eq:nfold-form}\\
    D_i\yy_i = \bb_i & \qquad \textrm{for all }i=1,2,\ldots,n,\nonumber
\end{align}
where again $C_i,D_i$ are integer $k\times k$ matrices and $\av,\bb_i$ are integer vectors of length $k$. This kind of programs appears for knapsack-like and scheduling problems, where blocks of variables $\yy_i$ correspond to some independent local decisions (for instance, whether to pack an item into the knapsack or not), and there are only a few linear constraints that involve all variables (for instance, that the capacity of the knapsack is not exceeded). See~\cite{ChenMYZ17,DeLoera2008,GavenciakKK22,JansenKMR22,KnopK18,KnopKLMO21,KnopKM20,KnopKM20b} for examples of $n$-fold programs appearing naturally ``in the wild''. \Cref{fig:blocks} depicts constraint matrices of two-stage stochastic and $n$-fold~programs.

\begin{figure}[h]
\centering
\begin{tikzpicture}
    \node at (-5,0) {$\begin{bmatrix}
        A_1 & D_1 & & & \\
        A_2 & & D_2 & & \\
        \vdots & & & \ddots & \\
        A_n & & & & D_n
    \end{bmatrix}$};

    \node at (0,0) {$\begin{bmatrix}
        C_1 & C_2  & \dots & C_n \\
        D_1 & & & \\
        & D_2 & & \\
        & & \ddots & \\
        & & & D_n
    \end{bmatrix}$};

    \node at (5,0) {$\begin{bmatrix}
        B & C_1 & C_2  & \dots & C_n \\
        A_1 & D_1 & & & \\
        A_2 & & D_2 & & \\
        \vdots & & & \ddots & \\
        A_n & & & & D_n
    \end{bmatrix}$};

\end{tikzpicture}

    \caption{Constraint matrices in two-stage stochastic, $n$-fold, and $4$-block integer programs, respectively. (The last kind will be discussed later.) Every block is a $k\times k$ matrix, where $k$ is the parameter. Empty spaces denote blocks of~zeroes.}\label{fig:blocks}
\end{figure}

Both for two-stage stochastic programs and for $n$-fold programs, we can consider two computational problems. The simpler {\em{feasibility problem}} just asks whether the given program has a {\em{solution}}: an evaluation of variables in nonnegative integers that satisfies all the constraints. In the harder {\em{(linear) optimization problem}}, we are additionally given an integer weight $w_x$ for every variable $x$ appearing in the program, and the task is to minimize $\sum_{x\colon \text{variable}} w_x\cdot x$ over all solutions.

Two-stage stochastic and $n$-fold programs have recently gathered significant interest in the theoretical community for two reasons. On one hand, it turns out that for multiple combinatorial problems, their natural integer programming formulations take either of the two forms. On the other hand, one can actually design highly non-trivial fixed-parameter algorithms for the optimization problem for both two-stage stochastic and $n$-fold programs; we will review them in a minute. Combining this two points yields a powerful algorithmic technique that allowed multiple new tractability results and running times improvements for various problems of combinatorial optimization; see~\cite{ChenMYZ17,GavenciakKK22,JansenKMR22,KnopK18,KnopK22,KnopKLMO21,KnopKM20,KnopKM20b,KouteckyZ20} for examples.

Delving more into technical details, if by $\Delta$ we denote the maximum absolute value of any entry in the constraint matrix, then the optimization problem for
\begin{itemize}[nosep]
    \item two-stage stochastic programs~\cref{eq:twostage-form} can be solved in time \mbox{$2^{\Delta^{\Oh(k^2)}}\cdot n\log^{\Oh(k^2)} n$}~\cite{CslovjecsekEPVW21}; and
    \item $n$-fold programs~\cref{eq:nfold-form} can be solved in time $(k\Delta)^{\Oh(k^3)}\cdot n\log^{\Oh(k^2)} n$~\cite{CslovjecsekEHRW21}.
\end{itemize}
The results above are in fact pinnacles of an over-a-decade-long sequence of developments, which gradually improved both the generality of the results and the running times~\cite{AschenbrennerH07,CslovjecsekEHRW21,CslovjecsekEPVW21,DeLoera2008,EisenbrandHK18,monster,HemmeckeOR13,JansenLR20,DBLP:journals/mp/Klein22,KouteckyLO18}, as well as provided complexity lower bounds~\cite{JansenKL23,KnopPW20}. We refer the interested reader to the monumental manuscript of Eisenbrand et al.~\cite{monster} which provides a comprehensive perspective on this research area.

We remark that the tractability results presented above can be also extended to the setting where the global-local block structure present in two-stage stochastic and $n$-fold programs can be iterated further, roughly speaking to trees of bounded depth. This leads to the study of integer programs with bounded {\em{primal}} or {\em{dual treedepth}}, for which analogous tractability results have been established. Since these notions will not be of interest in this work, we refrain from providing further details and refer the interested reader to the works relevant for this direction~\cite{BrianskiKKPS22,ChanCKKP22,ChenM18,CslovjecsekEHRW21,CslovjecsekEPVW21,EisenbrandHK18,monster,DBLP:journals/mp/Klein22,KleinR22,KnopPW20,KouteckyLO18}.

All the above-mentioned works, be it on two-stage stochastic or $n$-fold programs, or on programs of bounded primal or dual treedepth, crucially rely on one tool: the notion of the {\em{Graver basis}} of a matrix. Formal definition can be found in \cref{sec:Gravers}, but intuitively speaking, the Graver basis of a matrix $A$ consists of ``minimal steps'' within the lattice of integer points belonging to the kernel of $A$. Therefore, the maximum norm of an element of the Graver basis reflects the ``granularity'' of this lattice. And so, the two fundamental observations underlying all the discussed developments are the following:
\begin{itemize}[nosep]
    \item in two-stage stochastic matrices (or more generally, matrices of bounded primal treedepth), the $\ell_\infty$ norm of the elements of the Graver basis is bounded in terms of $k$ (the dimension of every block) and the maximum absolute value of any entry appearing in the matrix (see~\cite[Lemma~28]{monster}); and
    \item an analogous result holds for $n$-fold matrices (or more generally, matrices of bounded dual treedepth) and the $\ell_1$ norm (see~\cite[Lemma~30]{monster}).
\end{itemize}
Based on these observations, various algorithmic strategies, including augmentation frameworks~\cite{JansenLR20,KouteckyLO18} and proximity arguments~\cite{CslovjecsekEHRW21,CslovjecsekEPVW21,EisenbrandHK18}, can be employed to solve respective integer programs.

The drawback of the Graver-based approach is that it heavily relies on the assumption that all the entries of the input matrices are (parametrically) bounded. Indeed, the norms of the elements of the Graver basis are typically at least as large as the entries, so lacking any upper bound on the latter renders the methodology inapplicable. This is in stark contrast with the results on fixed-parameter tractability of integer programming parameterized by the number of variables~\cite{Dadush12,DadushER23,FrankTardos1987,kannan1987minkowski,Lenstra1983,ReisR23}, which rely on different tools and for which no bound on the absolute values of the entries is required. In fact, two-stage stochastic programs generalize programs with a bounded number of variables (just do not use variables~$\yy_i$), yet the current results for two-stage stochastic programs do {\em{not}} generalize those for integer programs with few variables, because they additionally assume a bound on the absolute values of the entries.

The goal of this paper is to understand to what extent one can efficiently solve two-stage stochastic and $n$-fold programs
while allowing large entries on input.

\paragraph*{Our contribution.} We prove that both the feasibility problem for two-stage stochastic programs and the optimization problem for uniform $n$-fold programs (that is, where $C_1=C_2=\ldots=C_n$) can be solved in fixed-parameter time when parameterized by the dimensions of the blocks and the maximum absolute value of any entry appearing in the diagonal blocks $D_i$. That is, we allow the entries of the global blocks ($A_i$ and $C_i$, respectively) to be arbitrarily large, and in the case of $n$-fold programs, we require that all blocks $C_i$ are equal. The statements below summarize our results. ($\|P\|$ denotes the total bitsize of a program~$P$.)

\begin{theorem}\label{thm:2stage-intro}
   The feasibility problem for two-stage stochastic programs $P$ of the form~\cref{eq:twostage-form} can be solved in time $f(k,\max_i \|D_i\|_\infty)\cdot \|P\|$ for a computable function $f$, where $k$ is the dimension of all the blocks $A_i,D_i$.
\end{theorem}

\begin{theorem}\label{thm:nfold-intro}
    The optimization problem for $n$-fold programs $P$ of the form~\cref{eq:nfold-form} that are uniform (that is, satisfy $C_1=\ldots=C_n$) can be solved in time $f(k,\max_i \|D_i\|_\infty)\cdot \|P\|^{\Oh(1)}$ for a computable function $f$, where $k$ is the dimension of all the blocks $C_i,D_i$.
\end{theorem}

We remark that in \cref{thm:2stage-intro} it is necessary to include $\max_i \|D_i\|_\infty$ among the parameters, because with no bound on the entries of the constraint matrix, the feasibility problem for two-stage stochastic programs becomes $\mathsf{NP}$-hard already for a constant $k$. An easy way to see it is via a straightforward reduction from an $\mathsf{NP}$-hard problem called \textsc{Good Simultaneous Approximation}~\cite{Lagarias85}, where the input consists of a rational vector $\av \in \Q^n$, a desired error bound $\epsilon \in \Q$, and an upper bound $N \in \Zn$, and the goal is to decide if there exists an integer multiplier $Q \in \{1, 2, \ldots, N\}$ such that $Q\av$ is $\epsilon$-close in the infinity norm to an integer vector, i.e., $\exists_{\bb \in \Z^n} \|Q\av - \bb\|_\infty \leq \epsilon$. In \Cref{sec:nphard} we give a different argument, reducing directly from 3-SAT, and showing that two-stage stochastic feasibility is actually \emph{strongly} $\mathsf{NP}$-hard for $k=16$. This rules out also any running time of the form $f(k) \cdot (\|P\|+\max_i \|D_i\|_\infty)^{\Oh(1)}$.

The uniformity condition in \cref{thm:nfold-intro} is also necessary (unless $\mathsf{P}=\mathsf{NP}$), as one can very easily reduce {\sc{Subset Sum}} to the feasibility problem for $n$-fold programs with $k=2$ and $D_i$ being $\{0,1\}$-matrices. Indeed, given an instance of {\sc{Subset Sum}} consisting of positive integers $a_1,\ldots,a_n$ and a target integer~$t$, we can write the following $n$-fold program on variables $y_1,\ldots,y_n,y_1',\ldots,y_n'\in \Zn$: $y_i+y_i'=1$ for all $i=1,\ldots,n$ and $\sum_{i=1}^n a_i y_i=t$. We remark that uniform $n$-fold programs are actually of the highest importance, as this form typically arises in applications. In fact, many of the previous works named such problems ``$n$-fold'', while our formulation~\cref{eq:nfold-form} would be called ``generalized $n$-fold''.

We also remark that the algorithm of \cref{thm:nfold-intro} does not use the assumption that the number of {\em{rows}} of matrix $C$ is bounded by $k$ (formally, in the precise, statement~\cref{thm:nfold-main}, we do not consider this number among parameters). However, we stress that the bound on the number of {\em{columns}} of $C$ is heavily exploited, which sets our approach apart from many of the previous works~\cite{CslovjecsekEHRW21,EisenbrandHK18,KouteckyLO18}.

Further, observe that \cref{thm:2stage-intro} applies only to the feasibility problem for two-stage stochastic programs. We actually do not know whether \cref{thm:2stage-intro} can be extended to the optimization problem as well, and we consider determining this an outstanding open problem. We will discuss its motivation in more details later on. Also, we remark that \cref{thm:2stage-intro} seems to be the first algorithm for feasibility of two-stage stochastic programs that achieves truly linear dependence of the running time on the total input size; the earlier algorithms of~\cite{CslovjecsekEPVW21,KouteckyLO18} had at least some additional polylogarithmic factors.

Finally, note that the algorithms of \Cref{thm:2stage-intro,thm:nfold-intro} are not strongly polynomial (i.e., the running time depends on the total bitsize of the input, rather than is counted in the number of arithmetic operations), while the previous algorithms of~\cite{CslovjecsekEHRW21,CslovjecsekEPVW21,EisenbrandHK18,KouteckyLO18} for the stronger parameterization are. At least in the case of \Cref{thm:2stage-intro}, this is justified, as the problem considered there generalizes integer programming parameterized by the number of variables, for which strongly polynomial FPT algorithms are not known.

Not surprisingly, the proofs of \Cref{thm:2stage-intro,thm:nfold-intro} depart from the by now standard approach through Graver bases; they are based on entirely new techniques, with some key Graver-based insight needed in the case of \cref{thm:nfold-intro}. In both cases, the problem is ultimately reduced to (mixed) integer programming with a bounded number of (integral) variables, and this allows us to cope with large entries on input. We expand the discussion of our techniques in \cref{sec:overview}, which contains a technical overview of the proofs.

\paragraph*{$4$-block programs.} Finally, we would like to discuss another motivation for investigating two-stage stochastic and $n$-fold programs with large entries, namely the open question about the parameterized complexity of {\em{$4$-block integer programming}}. $4$-block programs are programs in which the constraint matrix has the block-structured form depicted in the right panel of \Cref{fig:blocks}; note that this form naturally generalizes both two-stage stochastic and $n$-fold programs. It is an important problem in the area to determine whether the feasibility problem for $4$-block programs can be solved in fixed-parameter time when parameterized by the dimension of blocks $k$ and the maximum absolute value of any entry in the input matrix. The question was asked by Hemmecke et al.~\cite{HemmeckeKW14}, who proposed an XP algorithm for the problem. Improvements on the XP running time were reported by Chen et al.~\cite{ChenKXS20}, and FPT algorithms for special cases were proposed by Chen et al.~\cite{ChenCZ21}; yet no FPT algorithm for the problem in full generality is known so far. We remark that recently, Chen et al.~\cite{ChenCZ22} studied the complexity of $4$-block programming while allowing large entries in all the four blocks of the matrix. They showed that then the problem becomes $\mathsf{NP}$-hard already for blocks of constant dimension, and they discussed a few special cases that lead to tractability.

We observe that in the context of the feasibility problem for uniform $4$-block programs (i.e., with $A_i=A$ and $C_i=C$ for all $i=1,\ldots,n$), it is possible to emulate large entries within the global blocks $A,B,C$ using only small entries at the cost of adding a bounded number of auxiliary variables. This yields the following reduction, which we prove in \Cref{app:4block}.

\begin{restatable}{observation}{fourblock}\label{obs:4block}\RestateRemark
    Suppose the feasibility problem for uniform $4$-block programs can be solved in time $f(k,\Delta)\cdot \|P\|^{\Oh(1)}$ for some computable function $f$, where $k$ is the dimension of every block and $\Delta$ is the maximum absolute value of any entry in the constraint matrix. Then the feasibility problem for uniform $4$-block programs can be also solved in time $g(k,\max_i\|D_i\|_\infty)\cdot \|P\|^{\Oh(1)}$ for some computable function $g$ under the assumption that all the absolute values of the entries in matrices $A,B,C$ are bounded by $n$.
\end{restatable}

Consequently, to approach the problem of fixed-parameter tractability of $4$-block integer programming, it is imperative to understand first the complexity of two-stage stochastic and $n$-fold programming with large entries allowed in the global blocks. And this is precisely what we do in this work.

We believe that the next natural step towards understanding the complexity of $4$-block integer programming would be to extend \cref{thm:2stage-intro} to the optimization problem; that is, to determine whether optimization of two-stage stochastic programs can be solved in fixed-parameter time when parameterized by $k$ and $\max_i \|D_i\|_\infty$. Indeed, lifting the result from feasibility to the optimization problem roughly corresponds to adding a single constraint that links all the variables, and $4$-block programs differ from two-stage stochastic programs precisely in that there may be up to $k$ such additional linking constraints. Thus, we hope that the  new approach to block-structured integer programming presented in this work may pave the way towards understanding the complexity of solving $4$-block integer programs.

\paragraph*{Acknowledgements.} This research has been initiated during the trimester on Discrete Optimization at the Hausdorff Research Institute for Mathematics (HIM) in Bonn, Germany. We thank the organizers of the trimester for creating a friendly and motivating research environment. We also thank Eleonore Bach, Fritz Eisenbrand, and Robert Weismantel, for pointing us to the work of Aliev and Henk~\cite{AlievH10}. Finally, we thank an anonymous reviewer for suggesting a simpler proof of \cref{lem:bundling} as well as pointing out the possibility of obtaining the lower bound presented in \Cref{sec:nphard}.

\section{Overview}\label{sec:overview}

In this section, we provide a technical overview of our results aimed at presenting the main ideas and new conceptual contributions.

\subsection{Two-stage stochastic programming}

We start with an overview on the proof of \cref{thm:2stage-intro}. We will heavily rely on the combinatorics of integer and polyhedral cones, so let us recall basic definitions and properties.

\paragraph*{Cones.} Consider an integer  matrix $D$ with $t$ columns and $k$ rows. The {\em{polyhedral cone}} spanned by $D$ is the set $\realcone(D)\coloneqq\{D\yy\colon \yy\in \myRn^t\}\subseteq \myRn^k$, or equivalently, the set of all vectors in $\myRn^k$ expressible as nonnegative combinations of the columns of $D$. Within the polyhedral cone, we have the {\em{integer cone}} where we restrict attention to nonnegative integer combinations: $\intcone(D)\coloneqq\{D\yy\colon \yy\in \Zn^t\}\subseteq \Z^k$. Finally, the {\em{integer lattice}} is the set $\lattice(D)\coloneqq\{D\yy\colon \yy\in \Z^t\}\subseteq \Z^k$ which comprises all integer combinations of columns of $D$ with possibly negative coefficients.

Clearly, not every integer vector in $\realcone(D)$ has to belong to $\intcone(D)$. It is not even necessarily the case that $\intcone(D)=\realcone(D)\cap \lattice(D)$, as there might be vectors that can be obtained both as a nonnegative combination and as an integer combination of columns of $D$, but every such integer combination must necessarily contain negative coefficients. To see an example, note that in dimension~$k=1$, this is the Frobenius problem: supposing all entries of $D$ are positive integers, the elements of $\intcone(D)$ are essentially all nonnegative numbers divisible by the gcd (greatest common divisor) of the entries of $D$, except that for small numbers there might be some aberrations: a positive integer of order~$\Oh(\|D\|_\infty^2)$ may not be presentable as a nonnegative combination of the entries of~$D$, even assuming it is divisible by the gcd of the entries of $D$.

However, the Frobenius example suggests that the equality $\intcone(D)=\realcone(D)\cap \lattice(D)$ is almost true, except for aberrations near the boundary of $\realcone(D)$. We forge this intuition into a formal statement presented below that says roughly the following: if one takes a look at $\intcone(D)$ at a large scale, by restricting attention to integer vectors $\myvv\in \Z^k$ with fixed remainders of entries modulo some large integer $B$, then $\intcone(D)$ behaves like a polyhedron. In the following, for a positive integer~$B$ and a vector $\rr\in \{0,1,\ldots,B-1\}^k$, we let $\Lambda_\rr^B$ be the set of all vectors $\myvv\in \Z^k$ such that $\myvv\equiv \rr\bmod B$, which means $v_i\equiv r_i\bmod B$ for all $i\in \{1,\ldots,k\}$.

\begin{theorem}[Reduction to Polyhedral Constraints, see \cref{thm:polyhedron}]\label{thm:reduction-overview}
Let $D$ be an integer matrix with $t$ columns and $k$ rows. Then there exists a positive integer $B$, computable from $D$, such that for every ${\rr\in \{0,1,\ldots,B-1\}^k}$, there exists a polyhedron $\Qq_\rr$ such that
$$\Lambda_\rr^B\cap \intcone(D) = \Lambda_\rr^B\cap \Qq_\rr.$$
Moreover, a representation of such a polyhedron $\Qq_\rr$ can be computed given $D$ and $\rr$.
\end{theorem}

In other words, \cref{thm:reduction-overview} states that if one fixes the remainders of entries modulo $B$, then membership in the integer cone can be equivalently expressed through a finite system of linear inequalities. Before we sketch the proof of \cref{thm:reduction-overview}, let us discuss how to use this to solve two-stage stochastic programs.

\paragraph*{The algorithm.} Consider a two-stage stochastic program $P=(A_i,D_i,\bb_i\colon i\in \{1,\ldots,n\})$ such that blocks $A_i,D_i$ are integer $k\times k$ matrices and all entries of blocks $D_i$ are bounded in absolute value by~$\Delta$. The feasibility problem for $P$ can be understood as the question about satisfaction of the following sentence, where all quantifications range over $\Zn^k$:
\begin{equation}\label{eq:hipopotam}
    \exists_{\xx}\ \left(\bigwedge_{i=1}^n \exists_{\yy_i}\  A_i\xx+D_i\yy_i = \bb_i\right),\quad\textrm{or equivalently,}\quad \exists_{\xx}\ \left(\bigwedge_{i=1}^n\ \bb_i-A_i\xx\in \intcone(D_i)\right).
\end{equation}
Applying \cref{thm:reduction-overview} to each matrix $D_i$ yields a positive integer $B_i$. Note that there are only at most $(2\Delta+1)^{k^2}$ different matrices $D_i$ appearing in $P$, which also bounds the number of different integers $B_i$. By replacing all $B_i$s with their least common multiple, we may assume that $B_1=B_2=\ldots=B_n=B$. Note that $B$ is bounded by a computable function of $\Delta$ and $k$.

Consider a hypothetical solution $\xx,(\yy_i\colon i\in \{1,\ldots,n\})$ to $P$.
We guess, by branching into $B^k$ possibilities, a vector $\rr\in \{0,1,\ldots,B-1\}^k$ such that $\xx\equiv \rr\bmod B$. Having fixed $\rr$, we know how the vectors $\bb_i-A_i\xx$ look like modulo $B$, hence by \cref{thm:reduction-overview}, we may replace the assertion $\bb_i-A_i\xx\in \intcone(D_i)$ with the assertion $\bb_i-A_i\xx \in \Qq_{\rr_i}$, where $\rr_i\in \{0,1,\ldots,B-1\}^k$ is the unique vector such that $\bb_i-A_i\rr\equiv \rr_i\bmod B$. Thus, \eqref{eq:hipopotam} can be rewritten to the sentence
\begin{equation*}
    \bigvee_{\rr\in \{0,1,\ldots,B-1\}^k}
    \exists_{\xx}\ \left(\xx\equiv \rr\bmod B\right) \wedge  \left(\bigwedge_{i=1}^n \bb_i-A_i\xx\in \Qq_{\rr_i}\right),
\end{equation*}
which is equivalent to
\begin{equation}\label{eq:nosorozec}
    \bigvee_{\rr\in \{0,1,\ldots,B-1\}^k}
    \exists_{\xx}\exists_{\zz}\ \left(\xx=B\cdot \zz+\rr\right) \wedge  \left(\bigwedge_{i=1}^n \bb_i-A_i\xx\in \Qq_{\rr_i}\right).
\end{equation}
Verifying satisfiability of~\eqref{eq:nosorozec} boils down to solving $B^k$ integer programs on $2k$ variables $\xx$ and $\zz$ and linearly FPT many constraints, which can be done in linear fixed-parameter time using standard algorithms, for instance that of Kannan~\cite{kannan1987minkowski}.

We remark that the explanation presented above highlights that \cref{thm:reduction-overview} can be understood as a quantifier elimination result in the arithmetic theory of integers. This may be of independent interest, but we do not pursue this direction in this work.

\paragraph*{Reduction to polyhedral constraints.} We are left with sketching the proof of \cref{thm:reduction-overview}. Let $\Zz\coloneqq \Lambda_\rr^B\cap \intcone(D)$. Our goal is to understand that $\Zz$ can be expressed as the points of $\Lambda_\rr^B$ that are contained in some polyhedron $\Qq=\Qq_\rr$.

The first step is to understand $\realcone(D)$ itself as a polyhedron. This understanding is provided by a classic theorem of Weyl~\cite{Weyl35}: given $D$, one can compute a set of integer vectors $\Ff\subseteq \Z^k$ such that
$$\realcone(D)=\{\myvv\in \R^k~|~\sca{\ff}{\myvv}\geq 0\textrm{ for all }\ff\in \Ff\}.$$
Here, $\sca{\cdot}{\cdot}$ denotes the scalar product in $\R^k$. We will identify vectors $\ff\in \Ff$ with their associated linear functionals $\myvv\mapsto \sca{\ff}{\myvv}$. Thus, $\realcone(D)$ comprises all vectors $\myvv$ that have nonnegative evaluations on all functionals in $\Ff$. It is instructive to also think of the elements of $\Ff$ as of the facets of $\realcone(D)$ understood as a polyhedron, where the functional associated with $\ff\in \Ff$ measures the distance from the corresponding~facet.

Recall that in the context of \cref{thm:reduction-overview}, we consider vectors of $\Lambda_\rr^B$, that is, vectors $\myvv\in \Z^k$ such that $\myvv\equiv \rr\bmod B$. Then $\sca{\ff}{\myvv}\equiv \sca{\ff}{\rr}\bmod B$ for every $\ff\in \Ff$, hence we can find a unique integer $p_\ff\in \{0,1,\ldots,B-1\}$, $p_\ff\equiv \sca{\ff}{\rr}\bmod B$, such that $\sca{\ff}{\myvv}\equiv p_\ff\bmod B$ for all $\myvv\in \Lambda_\rr^B$. Now $\sca{\ff}{\myvv}$ is also nonnegative provided $\myvv\in \realcone(D)$, hence
$$ \sca{\ff}{\myvv}\in \{p_\ff,p_\ff+B,p_\ff+2B,\ldots\}\qquad\textrm{for all }\ff\in \Ff\textrm{ and } \myvv\in\Lambda_\rr^B\cap \realcone(D).$$
Now comes the key distinction about the behavior of $\myvv\in \Lambda_{\rr}^B\cap \realcone(D)$ with respect to $\ff\in \Ff$: we say that $\ff$ is {\em{tight}} with respect to $\myvv$ if $\sca{\ff}{\myvv}=p_\ff$, and is {\em{not tight}} otherwise, that is, if $\sca{\ff}{\myvv}\geq p_\ff+B$. Recall that in the context of \cref{thm:reduction-overview}, we are eventually free to choose $B$ to be large enough. Intuitively, this means that if $\ff$ is not tight for $\myvv$, then $\myvv$ lies far from the facet corresponding to $\ff$ and there is a very large slack in the constraint posed by $\ff$ understood as a functional. On the other hand, if $\ff$ is tight with respect to $\myvv$, then $\myvv$ is close to the boundary of $\realcone(D)$ at the facet corresponding to $\ff$, and there is a potential danger of observing Frobenius-like aberrations at $\myvv$.

Thus, the set $\Rr\coloneqq \Lambda_\rr^B\cap \realcone(D)$ can be partitioned into subsets $\{\Rr_\Gg\colon \Gg\subseteq \Ff\}$ defined as follows: $\Rr_{\Gg}$ comprises all vectors $\myvv\in \Rr$ such that $\Gg$ is exactly the set of functionals $\ff\in \Ff$ that are tight with respect to $\myvv$. Our goal is to prove that each set $\Rr_\Gg$ behaves uniformly with respect to $\Zz$: it is either completely disjoint or completely contained in $\Zz$. To start the discussion, let us look at the particular case of $\Rr_\Gg$ for $\Gg=\emptyset$. These are vectors that are deep inside $\realcone(D)$, for which no functional in $\Ff$ is tight. For these vectors, we use the following lemma, which is the cornerstone of our proof.

\begin{lemma}[Deep-in-the-Cone Lemma, simplified version of \cref{lem:cornetto}]\label{lem:dinc-overview}
 There exists a constant $M$, depending only on $D$, such that the following holds.
 Suppose $\myvv\in \realcone(D)\cap \Z^k$ is such that $\sca{\ff}{\myvv}>M$ for all $\ff\in \Ff$. Then $\myvv\in \intcone(D)$ if and only if $\myvv\in \lattice(D)$.
\end{lemma}
\begin{proof}
 The left-to-right implication is obvious, hence let us focus on the right-to-left implication. Suppose then that $\myvv\in \lattice(D)$.

 Let $\ww=\sum_{\dd\in D} L\cdot \dd$, where the summation is over the columns of $D$ and $L$ is a positive integer to be fixed later. Observe that for every $\ff\in \Ff$, we have $\sca{\ff}{\myvv-\ww}> M-L\cdot \sum_{\dd\in D}\sca{\ff}{\dd}$. Therefore, if we choose $M$ to be not smaller than $L\cdot \max_{\ff\in \Ff}\|\ff\|_1\cdot \|D\|_\infty$, then we are certain that $\sca{\ff}{\myvv-\ww}\geq 0$ for all $\ff\in \Ff$, and hence $\myvv-\ww\in \realcone(D)$. Consequently, we can write $\myvv-\ww=D\yy$ for some $\yy\in \myRn^t$. Let $\yy'\in \Zn^t$ be such that $y'_i=\lfloor y_i\rfloor$ for all $i\in \{1,\ldots,t\}$, and let $\myvv'=\ww+D\yy'$. Then
 $$\|\myvv-\myvv'\|_\infty = \|D(\yy-\yy')\|_\infty\leq t\cdot \|D\|_\infty.$$
 On the other hand, we clearly have $\myvv'\in \intcone(D)$ and by assumption, $\myvv\in \lattice(D)$. It follows that $\myvv-\myvv'\in \lattice(D)$. From standard bounds, see e.g.~\cite{Pottier91}, it follows that there exists $\zz\in \Z^t$ with $\myvv-\myvv'=D\zz$ such that $\|\zz\|_1$ is bounded by a function of $D$ and $\|\myvv-\myvv'\|_\infty$, which in turn is again bounded by a function of $D$ as explained above. (Note here that $t$ is the number of columns of $D$, hence it also depends only on $D$.) This means that if we choose $L$ large enough depending on $D$, we are certain that $\|\zz\|_1\leq L$. Now, it remains to observe that
 $$\myvv=\ww+D\yy'+(\myvv-\myvv')=D(L\cdot \boldsymbol{1}+\yy'+\zz),$$
 where $\boldsymbol{1}$ denotes the vector of $t$ ones, and that all the entries of $L\cdot \boldsymbol{1}+\yy'+\zz$ are nonnegative integers. This proves that $\myvv\in \intcone(D)$.
\end{proof}

We remark that the statement of Lemma~\ref{lem:dinc-overview} actually follows from results present in the literature, concerning the notion of {\em{diagonal Frobenius numbers}}. See the work of Aliev and Henk~\cite{AlievH10} for a broader discussion and pointers to earlier works, as well as the works of Aggrawal et al.~\cite{AggrawalJSW24} and of Bach et al.~\cite{BachERW24} for the newest developments on optimal bounds. As we will discuss in a moment, in this work we actually use a generalization of \cref{lem:dinc-overview}.

Consider any $\uu,\myvv\in \Rr$. Since all the entries of $\uu-\myvv$ are divisible by $B$, it is not hard to prove the following: if we choose $B$ to be a large enough factorial, then $\uu\in \lattice(D)$ if and only if $\myvv\in \lattice(D)$. Hence, from \cref{lem:dinc-overview} it follows that $\Rr_\emptyset$ is either entirely disjoint or entirely contained in $\Zz$.

A more involved reasoning based on the same fundamental ideas, but using a generalization of \cref{lem:dinc-overview}, yields the following lemma, which tackles also the case when some functionals of $\Ff$ are tight with respect to the considered vectors.

\begin{lemma}[see \cref{cl:tightImplies}]\label{lem:upwards-overview}
Suppose $\uu,\myvv\in \Rr$ are such that for every $\ff\in \Ff$, if $\ff$ is tight with respect to $\uu$, then $\ff$ is also tight with respect to $\myvv$. Then $\uu\in \Zz$ implies $\myvv\in \Zz$.
\end{lemma}

We remark that the proof of \cref{lem:upwards-overview} actually requires more work and more ideas than those presented in the proof of \cref{lem:dinc-overview}. In essence, one needs to partition functionals that are tight with respect to $\uu$ into those that are very tight (have very small $p_\ff$) and those that are only slightly tight (have relatively large $p_\ff$) in order to create a sufficient gap between very tight and slightly tight functionals. Having achieved this, a delicate variant of the reasoning from the proof of \cref{lem:dinc-overview} can be applied. It is important that whenever a functional $\ff\in \Ff$ is tight with respect to both $\uu$ and $\myvv$, we actually know that $\sca{\ff}{\uu}=\sca{\ff}{\myvv}=p_\ff$. Note that this is exactly the benefit achieved by restricting attention to the vectors of $\Lambda_\rr^B$.

Using \cref{lem:upwards-overview}, we can immediately describe how the structure of $\Zz$ relates to that of~$\Rr$.

\begin{corollary}[see \Cref{cl:allornothing}]\label{cor:structure}
For every $\Gg\subseteq \Ff$, either $\Rr_\Gg\cap \Zz=\emptyset$ or $\Rr_\Gg\subseteq \Zz$. Moreover, if $\Rr_\Gg\subseteq \Zz$ and $\Rr_{\Gg}$ is non-empty, then $\Rr_{\Gg'}\subseteq \Zz$ for all $\Gg'\subseteq \Gg$.
\end{corollary}

\cref{cor:structure} suggests now how to define the polyhedron $\Qq$. Namely, $\Qq$ is defined as the set of all $\myvv\in \R^k$ satisfying the following linear inequalities:
\begin{itemize}[nosep]
\item inequalities $\sca{\ff}{\myvv}\geq 0$ for all $\ff\in \Ff$ that define $\realcone(D)$; and
\item for every $\Gg\subseteq \Ff$ such that $\Rr_\Gg$ is nonempty and $\Rr_\Gg\cap \Zz=\emptyset$, the inequality
$$\sum_{\gg\in \Gg} \sca{\gg}{\myvv}\geq 1+\sum_{\gg\in \Gg} p_{\gg}.$$
\end{itemize}
In essence, the inequalities from the second point ``carve out'' those parts $\Rr_\Gg$ that should not be included in~$\Zz$. We note that computing the inequalities defining $\Qq$ requires solving several auxiliary integer programs to figure out for which $\Gg\subseteq \Ff$ the corresponding inequality should be included.

It is now straightforward to verify, using all the accumulated observations, that indeed $\Zz=\Rr\cap \Qq$ as required. This concludes a sketch of the proof of \cref{thm:reduction-overview}.

\subsection{\texorpdfstring{\nfold{}}{n-fold} programming}

We now give an overview of the proof of \cref{thm:nfold-intro}. For simplicity, we make the following assumptions.
\begin{itemize}[nosep]
\item We focus on the feasibility problem instead of optimization. At the very end, we will remark on what additional ideas are needed to also tackle the optimization problem.
\item We assume that all the diagonal blocks $D_i$ are equal: $D_i=D$ for all $i\in \{1,\ldots,n\}$, where $D$ is a $k\times k$ integer matrix with $\|D\|_\infty\leq \Delta$. This is only a minor simplification because there are only $(2\Delta+1)^{k^2}$ different matrices $D_i$ with $\|D_i\|_\infty\leq \Delta$, and in the general case, we simply treat every such possible matrix ``type'' separately using the reasoning from the simplified case.
\end{itemize}

\paragraph*{Breaking up bricks.} Basic components of the given $n$-fold program $P=(C,D,\av,\bb_i\colon i\in \{1,\ldots,n\})$ are {\em{bricks}}: programs $D\yy_i=\bb_i$ for $i\in \{1,\ldots,n\}$ that encode local constraints on the variables $\yy_i$. While the entries of $D$ are bounded in absolute values by the parameter $\Delta$, we do not assume any bound on the entries of vectors $\bb_i$. This poses an issue, as different bricks may have very different behaviors.

The key idea in our approach is to simplify the program $P$ by iteratively breaking up every brick $D\yy=\bb$ into two bricks $D\yy=\bb'$ and $D\yy=\bb''$ with strictly smaller right-hand sides $\bb',\bb''$, until eventually, we obtain an equivalent $n$-fold program $P'$ in which all right-hand sides have  $\ell_\infty$-norms bounded in terms of the parameters. The following lemma is the crucial new piece of technology used in our proof. (Here, we use the {\em{conformal order}} on $\Z^k$: we write $\uu\cfleq \myvv$ if $|\uu[i]|\leq |\myvv[i]|$ and $\uu[i]\cdot\myvv[i]\geq 0$ for all $i\in \{1,\ldots,k\}$.)

\begin{lemma}[Brick Decomposition Lemma, see \cref{lem:Decomp}]\label{lem:brick-overview}
There exists a function $g(k,\Delta)\in 2^{(k\Delta)^{\Oh(k)}}$ such that the following holds. Let $D$ be an integer matrix with $t$ columns and $k$ rows and all absolute values of its entries bounded by $\Delta$. Further, let $\bb\in \Z^{k}$ be an integer vector such that $\|\bb\|_\infty>g(k,\Delta)$. Then there are non-zero vectors $\bb',\bb''\in \Z^k$ such that:
\begin{itemize}[nosep]
    \item $\bb',\bb''\cfleq \bb$ and $\bb=\bb'+\bb''$; and
    \item for every $\myvv\in \Zn^{\yy}$ satisfying $D\myvv=\bb$, there exist $\myvv',\myvv''\in \Zn^{\yy}$ such that
    $$\myvv=\myvv'+\myvv'',\qquad D\myvv'=\bb',\qquad\textrm{and}\qquad  D\myvv''=\bb''.$$
\end{itemize}
\end{lemma}

In other words, \cref{lem:brick-overview} states that the brick $D\yy=\bb$ can be broken into two new bricks $D\yy'=\bb'$ and $D\yy''=\bb''$ with conformally strictly smaller $\bb',\bb''$ so that {\em{every}} potential solution $\myvv$ to $D\yy=\bb$ can be decomposed into solutions $\myvv',\myvv''$ to the two new bricks. It is easy to see that this condition implies that in $P$, we may replace the brick $D\yy=\bb$ with $D\yy'=\bb'$ and $D\yy''=\bb''$ without changing feasibility or, in the case of the optimization problem, the minimum value of the optimization goal. In the latter setting, both new bricks inherit the optimization vector $\cc_i$ from the original brick.

Before we continue, let us comment on the proof of \cref{lem:brick-overview}. We use two ingredients. The first one is the following fundamental result of Klein~\cite{DBLP:journals/mp/Klein22}. (Here, for a multiset of vectors $A$, by $\sum A$ we denote the sum of all the vectors in $A$.)

\begin{lemma}[Klein Lemma, variant from~\cite{CslovjecsekEPVW21}]\label{lem:klein-overview}
Let $T_1,\ldots,T_n$ be non-empty multisets of vectors in $\Z^k$ such that $\sum T_1=\sum T_2=\ldots=\sum T_n$ and all vectors contained in all multisets $T_1,\ldots,T_n$ have $\ell_\infty$-norm bounded by $\Delta$. Then there are non-empty multisets $S_1\subseteq T_1,\ldots,S_n\subseteq T_n$, each of size at most $2^{\Oh(k\Delta)^k}$, such that $\sum S_1=\sum S_2=\ldots=\sum S_n$.
\end{lemma}

In the context of the proof of \cref{lem:brick-overview}, we apply~\cref{lem:klein-overview} to the family of all multisets~$T$ that consist of columns of $D$ and satisfy $\sum T=\bb$. By encoding multiplicities, such multisets correspond to vectors $\myvv\in \Zn^k$ satisfying $D\myvv=\bb$. (We hide here some technicalities regarding the fact that this family is infinite.) By \cref{lem:klein-overview}, from each such multiset $T$, we can extract a submultiset $S$ of bounded size such that all the submultisets $S$ sum up to the same vector $\bb'$. Denoting $\bb''=\bb-\bb'$, this means that every vector $\myvv\in \Zn^k$ satisfying $D\myvv=\bb$ can be decomposed as $\myvv=\myvv'+\myvv''$ with $\myvv',\myvv''\in \Zn^k$ so that $D\myvv'=\bb'$ and $D\myvv''=\bb''$. Namely, $\myvv'$ corresponds to the vectors contained in $S$ and $\myvv''$ corresponds to the vectors contained in $T\setminus S$, where $T$ is the multiset corresponding to $\myvv$.

There is an issue in the above reasoning: we do not obtain the property $\bb',\bb''\cfleq \bb$, which will be important in later applications of \cref{lem:brick-overview}. To bridge this difficulty, we apply the argument above exhaustively to decompose $\bb$ as $\bb_1+\ldots+\bb_m$, for some integer $m$, so that every vector $\bb_i$ has the $\ell_\infty$-norm bounded by $2^{\Oh(k\Delta)^k}$ and every vector $\myvv\in \Zn^k$ satisfying $D\myvv=\bb$ can be decomposed as $\myvv=\myvv_1+\ldots+\myvv_m$ where $\myvv_i\in \Zn^k$ satisfies $D\myvv_i=\bb_i$. Then, we treat vectors $\bb_1,\ldots,\bb_m$ with the following lemma.

\begin{lemma}[see \cref{lem:bundling}]
    \label{lem:bundling-overview}
    Let $\uu_1,\ldots,\uu_m$ be vectors in $\Z^k$ of $\ell_\infty$-norm bounded by $\Xi$, and let $\bb=\sum_{i=1}^m \uu_i$. Then the vectors $\uu_1,\ldots,\uu_m$ can be grouped into non-empty groups $U_1,\ldots,U_\ell$, each of size at most $ \Oh(\Delta)^{2^{k-1}}$, so that $\sum U_i\cfleq \bb$ for all $i=1,\ldots,\ell$.
\end{lemma}

More precisely, \cref{lem:bundling-overview} allows us to group vectors $\bb_1,\ldots,\bb_m$ into groups of bounded size so that the sum within each group is sign-compatible with $\bb$. Assuming $\|\bb\|_\infty$ is large enough, there will be at least two groups. Then, any non-trivial partition of the groups translates into a suitable decomposition $\bb=\bb'+\bb''$ with $\bb',\bb''\cfleq \bb$.

The proof of \cref{lem:bundling-overview} is by induction on $k$ and uses arguments similar to standard proofs of Steinitz Lemma. This concludes a sketch of the proof of \cref{lem:brick-overview}.
\medskip

Once \cref{lem:brick-overview} is established, it is natural to use it iteratively: break $\bb$ into $\bb',\bb''$, then break $\bb'$ into two even smaller vectors, and so on. By applying the argument exhaustively, eventually we obtain a collection of vectors $\bb_1,\ldots,\bb_m\cfleq \bb$ such that $\bb=\bb_1+\ldots+\bb_m$,  $\|\bb_i\|_\infty\leq 2^{(k\Delta)^{\Oh(k)}}$ for all $i\in \{1,\ldots,m\}$, and every $\myvv\in \Zn^k$ satisfying $D\myvv=\bb$ can be decomposed as $\myvv=\myvv_1+\ldots+\myvv_m$ with $\myvv_i\in \Zn^k$ and $D\myvv_i=\bb_i$ for all $i\in \{1,\ldots,m\}$. We call such a collection a {\em{faithful decompostion}} of $\bb$ of {\em{order}} $2^{(k\Delta)^{\Oh(k)}}$.

There is an important technical caveat here. Observe that the size $m$ of a faithful decomposition of a right-hand side $\bb$ can be as large as $\Omega(\|\bb\|_1)$, which is {\em{exponential}} in the bitsize of the program~$P$. So we cannot hope to compute a faithful decomposition explicitly within the target time complexity. However, observe that all vectors $\bb_i$ in a faithful decomposition $\Bb$ are bounded in $\ell_\infty$-norm by $\Xi\coloneqq 2^{(k\Delta)^{\Oh(k)}}$, and there are only at most $(2\Xi+1)^k$ different such vectors. Therefore, $\Bb$ can be encoded by storing, for each vector $\bb'$ present in $\Bb$, the multiplicity of~$\bb'$ in~$\Bb$. Thus, describing $\Bb$ takes $2^{(k\Delta)^{\Oh(k)}}\cdot \log \|\bb\|_\infty$ bits.

With this encoding scheme in mind, we show that a faithful decomposition $\Bb$ of a given vector $\bb$ of order at most $\Xi$ can be computed in fixed-parameter time $f(\Delta,k)\cdot (\log \|\bb\|_\infty)^{\Oh(1)}$, for a computable function $f$. For this, we show that one can extract parts of the decomposition in ``larger chunks'', at each step reducing the $\ell_1$-norm of the decomposed vector by a constant fraction; this gives a total number of steps logarithmic in $\|\bb\|_1$. In each step, to extract the next large chunk of the decomposition, we use the fixed-parameter algorithm for optimization problems definable in Presburger arithmetic, due to Kouteck\'y and Talmon~\cite{KouteckyT21}. We remark that in our context, this tool could be also replaced by the fixed-parameter algorithm of Eisenbrand and Shmonin~\cite{DBLP:journals/mor/EisenbrandS08} for $\forall\exists$ integer programming.

\paragraph*{Reduction to (mixed) integer programming with few variables.} With faithful decompositions understood, we can compute, for every right-hand side $\bb_i$ part of $P$, a faithful decomposition $\{\bb^1_i,\ldots,\bb_i^{m_i}\}$ of $\bb_i$. This allows us to construct an equivalent (in terms of feasibility and optimization) $n$-fold program~$P'$ by replacing each brick $D\yy_i=\bb_i$ with bricks $D\yy_i^j=\bb_i^j$ for $j\in \{1,\ldots,m_i\}$. Thus, the program $P'$ has an exponential number of bricks, but can be computed and described concisely: all right-hand sides are bounded in the $\ell_\infty$-norm by at most $\Xi$, so for every potential right-hand side $\bb$, we just write the multiplicity in which $\bb$ appears in $P'$. We remark that such {\em{high-multiplicity encoding}} of $n$-fold integer programs has already been studied by Knop et al.~\cite{KnopKLMO23}.

For convenience, let $\RHS\coloneqq \{-\Xi,\ldots,\Xi\}^k$ be the set of all possible right-hand sides, and for $\bb\in \RHS$, by $\cntTypes[\bb]$ we denote the multiplicity of $\bb$ in $P'$.

It is now important to better understand the set of solutions to a single brick $D\yy=\bb$ present in $P'$. Here comes a key insight stemming from the theory of Graver bases: as (essentially) proved by Pottier~\cite{Pottier91}, every solution $\ww\in \Zn^k$ to $D\ww=\bb$ can be decomposed as $\ww=\whww+\gg_1+\ldots+\gg_\ell$, where
\begin{itemize}[nosep]
\item $\whww\in \Zn^k$ is a {\em{base solution}} that also satisfies $D\whww=\bb$, but $\|\whww\|_\infty$ is bounded by a function of $\Delta$ and~$\|\bb\|_\infty$, and
\item $\gg_1,\ldots,\gg_\ell\in \Zn^k$ are elements of the Graver basis of $D$.
\end{itemize}
Here, the {\em{Graver basis}} of $D$ consists of all conformally-minimal non-zero vectors $\gg$ satisfying $D\gg=\mathbf{0}$. In particular, it is known that the Graver basis is always finite and consists of vectors of $\ell_\infty$ norm bounded by $(2k\Delta+1)^k$~\cite{EisenbrandHK18}. The decomposition explained above will be called a {\em{Graver decomposition}} of $\ww$.

For $\bb\in \RHS$, let $\Base[\bb]$ be the set of all possible base solutions $\whww$ to $D\yy=\bb$. As $\|\bb\|_\infty\leq \Xi$ and $\Xi$ is bounded by a function of the parameters under consideration, it follows that $\Base[\bb]$ consists only of vectors of bounded $\ell_\infty$-norms, and therefore it can be efficiently constructed.

Having this understanding, we can write an integer program $M$ with few variables that is equivalent to $P'$. The variables are as follows:
\begin{itemize}[nosep]
    \item For every $\bb\in \RHS$ and $\whww\in \Base[\bb]$, we introduce a variable $\zeta^\bb_\whww\in \Zn$ that signifies how many times in total $\whww$ is used in the Graver decompositions of solutions to individual~bricks.
    \item For every nonnegative vector $\gg$ in the Graver basis of $D$, we introduce a variable $\delta_\gg\in \Zn$ signifying how many times in total $\gg$ appears in the Graver decompositions of solutions to individual bricks.
\end{itemize}
Note that since program $P'$ is uniform, the guessed base solutions and elements of the Graver basis can be assigned to {\em{any}} brick with the same effect on the linking constraints of $P'$. Hence, it suffices to verify the cardinalities and the total effect on the linking constrains of $P'$, yielding the following constraints of~$M$:
\begin{itemize}[nosep]
    \item the translated linking constraints: $\sum_{\bb\in \RHS}\sum_{\whww\in \Base[\bb]} \zeta^{\bb}_\whww\cdot C\whww+\sum_{\gg\in \Graver(D), \gg\geq 0} \delta_\gg\cdot C\gg =\av$.
    \item for every $\bb\in \RHS$, the cardinality constraint $\sum_{\whww\in \Base[\bb]}\zeta^\bb_\whww=\cntTypes[\bb]$.
\end{itemize}
Noting that the number of variables of $M$ is bounded in terms of the parameters, we may apply any fixed-parameter algorithm for integer programming parameterized by the number of variables, for instance that of Kannan~\cite{kannan1987minkowski}, to solve $M$. This concludes the description of the algorithm for the feasibility problem.

In the case of the optimization problem, there is an issue that the optimization vectors $\cc_i$ may differ between different bricks, and there may be as many as $n$ different such vectors. While the Graver basis elements can be always greedily assigned to bricks in which their contribution to the optimization goal is the smallest, this is not so easy for the base solutions, as every brick may accommodate only one base solution.
We may enrich $M$ by suitable assignment variables $\omega^{\bb,i}_{\whww}$ to express how many base solutions of each type are assigned to bricks with different optimization vectors; but this yields as many as $\Omega(n)$ additional variables. Fortunately, we observe that in the enriched program $M$, if one fixes any integral valuation of variables $\zeta^{\bb}_{\whww}$ and $\delta_\gg$, the remaining problem on variables $\omega^{\bb,i}_{\whww}$ corresponds to a flow problem, and hence its constraint matrix is totally unimodular. Thus, we may solve $M$ as a mixed integer program where variables $\omega^{\bb,i}_{\whww}$ are allowed to be fractional. The number of integral variables is bounded in terms of parameters, so we may apply the fixed-parameter algorithm for mixed integer programming of Lenstra~\cite{Lenstra1983}.

\section{Preliminaries}

We write $\Zn$ and $\myRn$ for the sets of nonnegative integers and nonnegative reals, respectively.

\paragraph*{Linear algebra.} In the technical part of this paper we choose to use a somewhat unorthodox notation for linear algebra, which we introduce now. The correspondence between this notation and the standard notation used in the previous sections is straightforward.

A {\em{tuple of variables}} is just a finite set of variable names. We use the convention that tuples of variables are denoted $\xx,\yy,\zz,$ etc., while individual variable names are $x,y,z,$ etc. For a set $A$ and a tuple of variables~$\xx$, an {\em{$\xx$-vector}} over $A$ is a function $\uu\colon \xx\to A$.
The value of $\uu$ on a variable $x\in \xx$ will be denoted by $\uu[x]$, and we call it the {\em{$x$-entry}} of $\uu$.
Vectors will be typically denoted by $\av,\bb,\cc,\uu,\myvv,\ww$ and so on. We write $A^{\xx}$ for the set of all $\xx$-vectors over $A$.
We can also apply arithmetics to vectors coordinate-wise in the standard manner, whenever the set $A$ is endowed with a structure of a ring. When $A=\R$, as usual we write
$\|\uu\|_\infty\coloneqq \max_{x\in \xx} |\uu[x]|$ and $\|\uu\|_1\coloneqq \sum_{x\in \xx} |\uu[x]|$.

For a ring $R$ (typically $\Z$ or $\R$) and tuples of variables $\tup x$ and $\tup y$, an {\em{$\tup x\times \tup y$-matrix}} is an $\tup x$-vector of $\tup y$-vectors, that is, an element of $(R^{\tup y})^{\tup x}$; we denote the latter set as $\matrices{R}{\xx}{\yy}$ for convenience.
The reader should think of the individual $\yy$-vectors comprised in a matrix as of the columns of the said matrix. We use capital letters $A,B,C,$ etc. for matrices. For $x\in \xx$ and $y\in \yy$, we write $A[x,y]\coloneqq (A[x])[y]$, that is, $A[x,y]$ is the $y$-entry of the column of $A$ corresponding to variable~$x$. As for vectors, by $\|A\|_{\infty}\coloneqq \max_{x\in \xx,y\in \yy} |A[x,y]|$ we denote the maximum absolute value of an entry of $A$.

For a matrix $A\in \matrices{R}{\xx}{\yy}$ and a vector $\uu\in R^{\xx}$, $A\uu$ is the vector $\myvv\in R^{\yy}$ satisfying
$$\myvv[y]=\sum_{x\in \xx} A[x,y]\cdot \uu[x]\qquad\textrm{for all }y\in \yy.$$
For vectors $\uu,\myvv\in R^{\xx}$, the {\em{inner product}} of $\uu$ and $\myvv$ is
$$\sca{\uu}{\myvv}\coloneqq \sum_{x\in \xx} \uu[x]\cdot \myvv[x].$$
All-zero and all-one vectors are denoted by $\zero$ and $\tup 1$, respectively, and their domains will be always clear from the context.

\paragraph*{Conformal order.} Let $\xx$ be a tuple of variables. Vectors $\uu,\myvv\in \Z^{\xx}$ are {\em{sign-compatible}} if $\uu[x]\cdot \myvv[x]\geq 0$ for all $x\in \xx$; that is, on every coordinate $\uu$ and $\myvv$ must have the same sign, where $0$ is assumed to be compatible with both signs.
The {\em{conformal order}} $\cfleq$ on vectors in $\Z^{\xx}$ is defined as follows: For two vectors $\uu,\myvv\in \Z^{\xx}$, we have $\uu\cfleq \myvv$ if $\uu$ and $\myvv$ are sign-compatible and
$$|\uu[x]|\leq |\myvv[x]|\qquad \textrm{for all }x\in \xx.$$
Note that thus, $\cfleq$ is a partial order on $\Z^{\xx}$.

\paragraph*{Collections.} A {\em{collection}} over a set $A$ is a function $C\colon I\to A$, where $I$ is the {\em{index set}} of $C$.
We often use notation $C=\mset{c_i\colon i\in I}$ to enumerate the elements of a collection with the elements of its index set, where formally $c_i$ is the value of $C$ on $i$. Note that syntactically, a collection with index set $I$  is just an $I$-vector, but we will use collections and vectors in different ways. Typically, when speaking about collections we are not really interested in the index set, and the reader can always assume it to be a prefix of natural numbers. However, in notation it will be convenient to consider various index sets.

Standard set theory notation can be used in the context of  collections in the natural manner, by applying it to the index sets. In particular, a {\em{subcollection}} of a collection $C=\mset{c_i\colon i\in I}$ is any collection $C'=\mset{c_i\colon i\in I'}$ for $I'\subseteq I$.

\paragraph*{Block-structured integer programming.} We now introduce the two variants of block-structured integer programming problems that we are interested in: \stage{} integer programming and \nfold{} integer programming.

A {\em{\stage{} integer program}} consists of three collections:
\begin{itemize}[nosep]
 \item $\mset{A_i\colon i\in I}$ is a collection of matrices in $\matrices{\Z}{\xx}{\tv}$;
 \item $\mset{D_i\colon i\in I}$ is a collection of matrices in $\matrices{\Z}{\yy}{\tv}$; and
 \item $\mset{\bb_i\colon i\in I}$ is a collection of vectors in $\Z^{\tv}$.
\end{itemize}
Here, $\xx,\yy,\tv$ are tuples of variables and $I$ is the index set of the program. A {\em{solution}} to such a program consists of a vector $\uu\in \Zn^\xx$ and vectors $\mset{\myvv_i\colon i\in I}$ in $\Zn^{\yy}$ such that
\begin{gather*}
A_i\uu+D_i\myvv_i = \bb_i \quad \textrm{for all }i\in I.
\end{gather*}
In the \stageFeas{} problem, we are given a \stage{} integer program in the form described above, and the task is to decide whether this program has a solution.


An {\em{\nfold{} integer program}} consists of the following components:
\begin{itemize}[nosep]
 \item $\mset{C_i\colon i\in I}$ is a collection of matrices in $\matrices{\Z}{\yy}{\sv}$;
 \item $\mset{D_i\colon i\in I}$ is a collection of matrices in $ \matrices{\Z}{\yy}{\tv}$;
 \item $\av$ is a vector in $\Z^{\sv}$; and
 \item $\mset{\bb_i\colon i\in I}$ is a collection of vectors in $\Z^{\tv}$.
\end{itemize}
Again, $\yy,\sv,\tv$ are tuples of variables and $I$ is the index set of the program. A {\em{uniform}} program is one where $C_i=C_j$ for all $i,j\in I$.
A {\em{solution}} to the program described above consists of vectors $\mset{\myvv_i\colon i\in I}$ in $\Zn^\yy$ such that
$$\sum_{i\in I}C_i\myvv_i=\av\qquad\textrm{and}\qquad D_i\myvv_i=\bb_i\quad\textrm{for all }i\in I.$$
The \nfoldFeas{} problem asks for the existence of a solution to a given \nfold{} integer program, while in the \nfoldOpt{} problem, the task is to minimize $\sum_{i\in I}\sca{\cc_i}{\myvv_i}$ among the solutions, for additionally given vectors $\mset{\cc_i\colon i\in I}$ in~$\Z^\yy$. A prefix {\sc{Uniform}} may be added to the problem name to specify that we speak about uniform programs. (Note that only the matrices $C_i$ have to be equal, and not necessarily the optimization goal vectors $\cc_i$.)

\medskip

For an integer program $P$, be it a \stage{} or an \nfold{} program, by $\|P\|$ we denote the total bitsize of the encoding of $P$, where all numbers are encoded in binary.
\section{Two-stage stochastic integer programming}

Our main result for \stage{} integer programs is captured in the following statement, which is \cref{thm:2stage-intro} with adjusted notation.

\begin{theorem}\label{thm:2stage-main}
An instance $P=\mset{A_i,D_i,\bb_i\colon i\in I}$ of \stageFeas{} can be solved in time $f(\Delta,|\xx|,|\tv|)\cdot \|P\|$ for some computable function $f$, where $\Delta=\max_{i\in I} \|D_i\|_\infty$.
\end{theorem}

Note that importantly, the absolute values of the entries of matrices $A_i$ are {\em{not}} assumed to be bounded in terms of the considered parameters. Also, in \cref{thm:2stage-main} the parameters do not include $|\yy|$, the  number of columns of each block $D_i$. This is because by removing equal columns in those blocks (which does not affect the feasibility of the program), we may always assume $|\yy|\leq (2\Delta+1)^{|\tv|}$.

The proof of \cref{thm:2stage-main} spans the entirety of this section.

\subsection{Cones}

In our proof we will rely on the geometry of polyhedral and integer cones. For this, we introduce the following notation. Let $\tv$ be a tuple of variables and $\Dd\subseteq \Z^\tv$ be a finite set of vectors.
For $S\in \{\R,\myRn,\Z,\Zn\}$, we define the {\em{$S$-span of $\Dd$}} as
$$\spn^S(\Dd)\coloneqq \left\{\sum_{\dd\in \Dd} \lambda_\dd\cdot  \dd~\colon~\mset{\lambda_\dd\colon \dd\in \Dd} \textrm{ is a collection of elements of }S\right\}\subseteq \R^{\tv}.$$
In other words, $\spn^S(\Dd)$ are all vectors that can be expressed as a linear combination of vectors in $\Dd$ with coefficients belonging to $S$.
For clarity, we write
\begin{align*}
\spn(\Dd)&\coloneqq \spn^{\R}(\Dd), &\qquad \realcone(\Dd)&\coloneqq \spn^{\myRn}(\Dd),\\
\lattice(\Dd)&\coloneqq \spn^{\Z}(\Dd), &\qquad \intcone(\Dd)&\coloneqq \spn^{\Zn}(\Dd).
\end{align*}
Thus, $\spn(\Dd)$ is just the usual linear space spanned of $\Dd$. Next, $\realcone(\Dd)$ is the standard {\em{polyhedral cone}}: it consists of all vectors expressible as nonnegative linear combinations of vectors in $\Dd$. Similarly, $\intcone(\Dd)$ is the {\em{integer cone}}, where we restrict attention to nonnegative integer combinations of vectors in $\Dd$. Finally, $\lattice(\Dd)$ is the {\em{integer lattice}} consisting of all vectors that can be reached from $\zero$ by adding and subtracting vectors of $\Dd$.

We remark that while the reader may think of $\Dd$ as of a matrix whose columns are the vectors of~$\Dd$, ordered arbitrarily, in the notation we will consistently treat $\Dd$ as a set of vectors. Consequently, $|\Dd|$ denotes the number of vectors in $\Dd$ (or, the number of columns of $\Dd$ treated as a matrix), and the same also applies to all subsets of $\Dd$.

\subsubsection{Dual representation of cones}

We will use the classic result of Weyl about representing polyhedral cones as intersections of half-spaces.

\begin{theorem}[Weyl,~\cite{Weyl35}]\label{thm:weyl}
 For every set of vectors $\Dd\subseteq \Z^\tv$ there exists a finite set of vectors $\Ff\subseteq \Z^\tv$ such~that
 $$\realcone(\Dd)=\{\myvv\in \R^{\tv}~|~\sca{\ff}{\myvv}\geq 0\textrm{ for all }\ff\in \Ff\}.$$
 Moreover, such a set $\Ff$ can be computed given $\Dd$.
\end{theorem}

A set $\Ff$ satisfying the outcome of \cref{thm:weyl} will be called a {\em{dual representation}} of $\realcone(\Dd)$. Here is a simple observation about how elements of the dual representation relate to the elements of $\Dd$.

\begin{observation}\label{obs:dualPositive}
 Suppose $\Ff$ is a dual representation of $\realcone(\Dd)$ for a set of vectors $\Dd\subseteq \Z^\tv$. Then
 $$\sca{\ff}{\dd}\geq 0\quad \textrm{ for all }\ \ff\in \Ff \textrm{ and }\ \dd\in \Dd.$$
\end{observation}
\begin{proof}
 Clearly $\dd\in \realcone(\Dd)$, so $\sca{\ff}{\dd}\geq 0$ because $\Ff$ is a dual representation of $\realcone(\Dd)$.
\end{proof}

For a subset $\Gg\subseteq \Ff$, we define
$$\Dd_\Gg \coloneqq \{\dd\in \Dd~|~\sca{\gg}{\dd}=0\textrm{ for all }\gg\in \Gg\}.$$
In other words, $\Dd_\Gg$ consists of those vectors of $\Dd$ that are orthogonal to all vectors in $\Gg$. The next lemma expresses the following intuition: the facet of $\realcone(\Dd)$ corresponding to $\Gg$ is spanned by the vectors of~$\Dd_\Gg$.

\begin{lemma}\label{lem:facetRepresentation}
 Suppose $\Ff$ is a dual representation of $\realcone(\Dd)$ for a set of vectors $\Dd\subseteq \Z^\tv$. Then for every subset $\Gg\subseteq \Ff$ we have
 $$\realcone(\Dd_\Gg) = \realcone(\Dd)\cap \{\myvv\in \R^{\tv}~|~\sca{\gg}{\myvv}=0\textrm{ for all }\gg\in \Gg\}.$$
\end{lemma}
\begin{proof}
 Note that $\Dd_\Gg\subseteq \Dd$ entails $\realcone(\Dd_\Gg)\subseteq \realcone(\Dd)$. Further, since $\sca{\gg}{\dd}=0$ for all $\gg\in \Gg$ and all $\dd\in \Dd_\Gg$, the same also holds for all linear combinations of vectors in $\Dd_\Gg$, implying that $\sca{\gg}{\myvv}=0$ for all $\gg\in \Gg$ and $\myvv\in \realcone(\Dd_\Gg)$. This proves inclusion $\subseteq$.

 For the converse inclusion, consider any $\myvv\in \realcone(\Dd)$ such that $\sca{\gg}{\myvv}=0$ for all $\gg\in \Gg$. As $\myvv\in \realcone(\Dd)$, there is a collection of nonnegative reals $\mset{\lambda_\dd\colon \dd\in \Dd}$ such that
 $$\myvv = \sum_{\dd\in \Dd} \lambda_\dd\cdot \dd.$$
 Consider any $\ee\in \Dd\setminus \Dd_\Gg$. By the definition of $\Dd_\Gg$ and \cref{obs:dualPositive}, there exists $\gg\in \Gg$ such that $\sca{\gg}{\ee}>0$. As all coefficients $\mset{\lambda_\dd\colon \dd\in \Dd}$ are nonnegative, by \cref{obs:dualPositive} we have
 $$0=\sca{\gg}{\myvv}=\sum_{\dd\in \Dd} \lambda_\dd\cdot \sca{\gg}{\dd} \geq \lambda_\ee\cdot \sca{\gg}{\ee}.$$
 As $\sca{\gg}{\ee}>0$, we necessarily have $\lambda_\ee=0$; and this holds for every $\ee\in \Dd\setminus \Dd_\Gg$. So in fact $\myvv$ is a nonnegative linear combination of vectors in $\Dd_\Gg$, or equivalently $\myvv\in \realcone(\Dd_\Gg)$. This proves inclusion~$\supseteq$.
\end{proof}

\subsubsection{Membership in integer cones and in integer lattices}

First, we need a statement that membership in an integer lattice can be witnessed by a vector with small entries. The following lemma follows from standard bounds given by Pottier~\cite[Corollary~4]{Pottier91} (c.f. \cref{lem:GraverDecomp}).

\begin{lemma}\label{lem:latticeSmallSol}
For every set of vectors $\Dd\subseteq \Z^\tv$ and a vector $\myvv\in \lattice(\Dd)$, there exist integer coefficients $\mset{\lambda_\dd\colon \dd\in \Dd}$ such that
$$\myvv=\sum_{\dd\in \Dd} \lambda_\dd\cdot \dd\qquad \textrm{and}\qquad \sum_{\dd\in \Dd} |\lambda_\dd|\leq \left(2+\max_{\dd\in \Dd} \|\dd\|_\infty+\|\myvv\|_\infty\right)^{2|\tv|}.$$
\end{lemma}

We now proceed to a technical statement that is crucial in our approach. Intuitively, it says that if a vector lies ``deep'' in the polyhedral cone, then its membership in the integer cone is equivalent to the membership in the integer lattice.

\begin{lemma}[Deep-in-the-Cone Lemma]\label{lem:cornetto}
 Let $\Dd\subseteq \Z^\tv$ be a set of vectors and $\Ff\subseteq \Z^\tv$ be a dual representation of $\realcone(\Dd)$. Then there is a positive integer $M$, computable from $\Dd$ and $\Ff$, such that for every $\Gg\subseteq \Ff$ and every $\myvv\in \realcone(\Dd_\Gg)$ satisfying $\sca{\ff}{\myvv}\geq M$ for all $\ff\in \Ff\setminus \Gg$, we have
 $$\myvv\in \lattice(\Dd_\Gg)\quad\textrm{if and only if}\quad \myvv\in \intcone(\Dd_\Gg).$$
\end{lemma}
\begin{proof}
 We set
 $$M\coloneqq L\cdot |\Dd|\cdot \max_{\ff\in \Ff}\|\ff\|_1\cdot \max_{\dd\in \Dd} \|\dd\|_\infty,\qquad \textrm{where}\qquad L\coloneqq \left(2+(|\Dd|+1)\cdot \max_{\dd\in \Dd} \|\dd\|_\infty\right)^{2|\tv|}$$
 is the bound provided by \cref{lem:latticeSmallSol} for vectors of $\ell_\infty$-norm bounded by $|\Dd|\cdot \max_{\dd\in \Dd} \|\dd\|_\infty$.

 As $\intcone(\Dd_\Gg)\subseteq \lattice(\Dd_\Gg)$, it suffices to prove the following: if $\myvv\in \realcone(\Dd_\Gg)\cap \lattice(\Dd_\Gg)$ satisfies $\sca{\ff}{\myvv}\geq M$ for all $\ff\in \Ff\setminus \Gg$, then in fact $\myvv\in \intcone(\Dd_\Gg)$. Let
 $$\ww\coloneqq \sum_{\dd\in \Dd_\Gg} L\cdot \dd.$$
 As both $\myvv$ and $\ww$ are linear combinations of vectors in $\Dd_\Gg$, we have
 $$\sca{\gg}{\myvv-\ww}=0\qquad\textrm{for all }\gg\in \Gg.$$
 Further, for each $\ff\in \Ff\setminus \Gg$ we have
 $$\sca{\ff}{\myvv-\ww}=\sca{\ff}{\myvv}-\sca{\ff}{\ww}\geq M-L\cdot \sum_{\dd\in \Dd} |\sca{\ff}{\dd}|\geq M-L\cdot |\Dd|\cdot \max_{\ff\in \Ff}\|\ff\|_1\cdot \max_{\dd\in \Dd} \|\dd\|_\infty=0.$$
 As $\Ff$ is a dual representation of $\realcone(\Dd)$, the two assertions above together with \cref{lem:facetRepresentation} imply that $\myvv-\ww\in \realcone(\Dd_\Gg)$. Consequently, there is a collection of nonnegative reals $\mset{\lambda_\dd\colon \dd\in \Dd_\Gg}$ such that
 $$\myvv-\ww=\sum_{\dd\in \Dd_\Gg} \lambda_\dd\cdot \dd.$$
 Now, consider the vector
 $$\myvv'\coloneqq \ww+\sum_{\dd\in \Dd_\Gg} \lfloor\lambda_\dd\rfloor\cdot \dd=\sum_{\dd\in \Dd_\Gg} (L+\lfloor\lambda_\dd\rfloor)\cdot \dd.$$
 Clearly $\myvv'\in \lattice(\Dd_\Gg)$. As $\myvv\in \lattice(\Dd_\Gg)$ by assumption, we have $\myvv-\myvv'\in \lattice(\Dd_\Gg)$ as well. Furthermore,
 $$\|\myvv-\myvv'\|_\infty = \left\|\sum_{\dd\in \Dd_\Gg} (\lambda_\dd-\lfloor \lambda_\dd\rfloor)\cdot \dd\right\|_\infty\leq |\Dd_\Gg|\cdot \max_{\dd\in \Dd_\Gg} \|\dd\|_\infty\leq |\Dd|\cdot \max_{\dd\in \Dd} \|\dd\|_\infty.$$
 By \cref{lem:latticeSmallSol} we conclude that there exist integer coefficients $\mset{\mu_\dd\colon \dd\in \Dd_\Gg}$ such that
 $$\myvv-\myvv'=\sum_{\dd\in \Dd_\Gg} \mu_\dd\cdot \dd\qquad \textrm{and}\qquad \sum_{\dd\in \Dd_\Gg} |\mu_\dd|\leq L.$$
 Hence,
 $$\myvv=\myvv'+(\myvv-\myvv')=\sum_{\dd\in \Dd_\Gg} (L+\lfloor\lambda_\dd\rfloor+\mu_\dd)\cdot \dd.$$
 It now remains to observe that each coefficient $L+\lfloor\lambda_\dd\rfloor+\mu_\dd$ is a nonnegative integer, because $\lambda_\dd\geq 0$ and $|\mu_\dd|\leq L$. This shows that $\myvv\in \intcone(\Dd_\Gg)$, thereby completing the proof.
\end{proof}

\subsection{Regular lattices}

As mentioned in \cref{sec:intro}, the key idea in our approach is to guess the remainders of the entries of the sought solution modulo some large integer. To describe this idea formally, we consider regular lattices defined as follows.

Let $K$ be a positive integer. For a vector of residues $\rr\in \{0,1,\ldots,K-1\}^\tv$, where $\tv$ is a tuple of variables, we define the {\em{regular lattice}} $\Lambda^K_{\rr}$:
$$\Lambda^K_{\rr}\coloneqq \{\myvv\in \Z^{\tv}~|~\myvv(t)\equiv \rr(t)\bmod K\textrm{ for all }t\in \tv\}.$$
In other words, $\Lambda^K_\rr$ comprises all integer vectors in which the remainders mod $K$ on all coordinates are exactly as specified in $\rr$. Note that regular lattices are affine rather than linear. In particular, $\zero\in \Lambda^K_\rr$ if and only if $\rr=\zero$, and hence a regular lattice $\Lambda^K_\rr$ is not the form $\lattice(\Dd)$ for a multiset of vectors $\Dd$, unless $\rr=\zero$.

We will need the following simple claim about fractionality of solutions to systems of equations.

\begin{lemma}\label{lem:fractionality}
 Let $\Dd\subseteq \Z^\tv$ be a set of vectors, where $\tv$ is a tuple of variables. Then there is a positive integer~$C$, computable from $\Dd$, satisfying the following: for every $\myvv\in \spn(\Dd)\cap \Z^\tv$ there exist coefficients $\mset{\lambda_\dd\colon \dd\in \Dd}$ such that
 $$\myvv=\sum_{\dd\in \Dd} \lambda_\dd\cdot \dd\qquad \textrm{and}\qquad C\lambda_\dd\textrm{ is an integer for every }\dd\in \Dd.$$
\end{lemma}
\begin{proof}
 Let $\zz$ be a tuple of variables containing one variable $z_\dd$ for each $\dd\in \Dd$, and let $D$ be the $\zz\times \tv$ matrix such that the column of $D$ corresponding to $z_\dd\in \zz$ is $\dd$. Thus, collections of coefficients $\mset{\lambda_\dd\colon \dd\in \Dd}$ as in the lemma statement correspond to solutions $\boldsymbol{\lambda}\in \Z^\zz$ of the system of equations $D\zz=\myvv$. By first restricting the rows of $D$ to a linearly independent set of rows, and then adding some $\{0,1\}$ row vectors together with zeroes on the right hand side, we can obtain a system of equations $\widetilde{D}\zz=\widetilde{\myvv}$ such that every solution to $\widetilde{D}\zz=\widetilde{\myvv}$ is also a solution to $D\zz=\myvv$, and $\widetilde{D}$ is a non-singular square matrix. It now follows from Cramer rules that if $\boldsymbol{\lambda}\in \Z^\zz$ is the unique solution to $\widetilde{D}\zz=\widetilde{\myvv}$, then $\det \widetilde{D}\cdot {\boldsymbol{\lambda}}$ is an integer. Therefore, we may set $C\coloneqq \det \widetilde{D}$, where $\widetilde{D}$ is any non-singular square matrix that can be obtained from $D$ by first restricting it to a linearly independent set of rows, and then extending this set of rows to a row basis using $\{0,1\}$ row vectors. Note that $C$ is clearly computable from $D$.
\end{proof}

We now use \cref{lem:fractionality} to prove the following statement that relates regular lattices with integer lattices generated by sets of vectors.

\begin{lemma}\label{lem:spaceCramer}
 Let $\Dd\subseteq \Z^\tv$ be a set of vectors, where $\tv$ is a tuple of variables. Then there is a positive integer~$K$, computable from $\Dd$, satisfying the following: for every positive integer $K'$ divisible by $K$ and every vector of residues $\rr\in \{0,1,\ldots,K'-1\}^\tv$,
 $$\spn(\Dd)\cap \Lambda^{K'}_\rr \subseteq \lattice(\Dd)\qquad\textrm{or}\qquad \lattice(\Dd)\cap \Lambda^{K'}_\rr=\emptyset.$$
\end{lemma}
\begin{proof}
 We set $K\coloneqq C$, where $C$ is the integer provided by \cref{lem:fractionality} for $\Dd$. Fix any positive integer $K'$ divisible by $K$ and $\rr\in \{0,1,\ldots,K'-1\}^\tv$.
 It suffices to prove that if there exists $\uu\in \lattice(\Dd)\cap \Lambda^{K'}_\rr$, then in fact $\spn(\Dd)\cap \Lambda^{K'}_\rr \subseteq \lattice(\Dd)$. Consider any $\myvv\in \spn(\Dd)\cap \Lambda^{K'}_\rr$. Since $\uu,\myvv\in \Lambda^{K'}_\rr$, every entry of the vector $\uu-\myvv$ is divisible by $K'$, so also by $K$. It follows that
 $$\uu-\myvv = K\cdot \ww\quad \textrm{for some }\ww\in \Z^{\tv}.$$
 Since $\uu,\myvv\in \spn(\Dd)$, we also have $\ww\in \spn(\Dd)$. By \cref{lem:fractionality}, there exist coefficients $\mset{\lambda_\dd\colon \dd\in \Dd}$ such that
 $$\ww=\sum_{\dd\in \Dd} \lambda_\dd\cdot \dd\qquad \textrm{and}\qquad K\lambda_\dd\textrm{ is an integer for all }\dd\in \Dd.$$
 Therefore, we have
 $$\uu-\myvv = K\cdot \ww = \sum_{\dd\in \Dd} (K\lambda_\dd)\cdot \dd,$$
 where coefficients $K\lambda_\dd$ are integral. So $\uu-\myvv\in \lattice(\Dd)$. As $\uu\in \lattice(\Dd)$ by assumption, we conclude that $\myvv\in \lattice(\Dd)$ as well; this concludes the proof.
\end{proof}

\subsection{Reduction to polyhedral constraints}

We proceed to the cornerstone of our approach: the theorem stated below, which was presented in \cref{sec:overview} as \cref{thm:reduction-overview}. Intuitively, it says that under fixing residues modulo a large integer, membership in an integer cone is equivalent to membership in a carefully crafted polyhedron.

\begin{theorem}[Reduction to Polyhedral Constraints]\label{thm:polyhedron}
 Let $\Dd\subseteq \Z^\tv$ be a set of vectors, where $\tv$ is a tuple of variables. Then there exists a positive integer $B$, computable from $\Dd$, satisfying the following: for every vector of residues $\rr\in \{0,1,\ldots,B-1\}^\tv$, there is a finite set $\Qq\subseteq \Z^\tv\times \Z$ such that
 $$\Lambda^B_\rr\cap \intcone(\Dd) = \Lambda^B_\rr\cap \left\{\myvv\in \R^\tv~|~\sca{\qq}{\myvv}\geq a\textrm{ for all }(\qq,a)\in \Qq\right\}.$$
 Moreover, there is an algorithm that given $\Dd$ and $\rr$, computes $\Qq$ satisfying the above.
\end{theorem}
\begin{proof}
 Let $\Ff$ be a dual representation of $\realcone(\Dd)$, computed from $\Dd$ using \cref{thm:weyl}. Further, let
 \begin{itemize}[nosep]
  \item $K$ be the least common multiple of all integers $K_\Gg$ obtained by applying \cref{lem:spaceCramer} to $\Dd_\Gg$ for every $\Gg\subseteq \Ff$, and
  \item $M$ be the integer obtained by applying \cref{lem:cornetto} to $\Dd$ and $\Ff$.
 \end{itemize}
 We define
 $$\widehat{M}\coloneqq M+|\Dd|\cdot \max_{\ff\in \Ff} \|\ff\|_1\cdot \max_{\dd\in \Dd} \|\dd\|_\infty,$$
and set
 $$B_0\coloneqq K,\qquad B_i\coloneqq \widehat{M}\cdot B_{i-1}\quad\textrm{for }i=1,2,\ldots,|\Ff|,\qquad \textrm{and}\qquad B\coloneqq 2\cdot B_{|\Ff|}.$$
We are left with constructing a suitable set $\Qq$ for a given $\rr\in \{0,1,\ldots,B-1\}^{\tv}$.

For convenience, denote
$$\Rr\coloneqq \Lambda_\rr^B\cap \realcone(\Dd)\qquad\textrm{and}\qquad \Zz\coloneqq \Lambda_\rr^B\cap \intcone(\Dd).$$
For every $\ff\in \Ff$, let $p_\ff$ be the unique integer in $\{0,1,\ldots,B-1\}$ such that
$$p_\ff\equiv \sca{\ff}{\rr}\mod B.$$
Observe the following.

\begin{techclaim}\label{cl:discreteSteps}
 For every $\ff\in \Ff$ and $\myvv\in \Rr$, we have
 $$\sca{\ff}{\myvv}\in \{p_\ff,p_\ff+B,p_\ff+2B,\ldots\}.$$
\end{techclaim}
\begin{subproof}
 As $\myvv\in \realcone(\Dd)$ and $\Ff$ is a dual representation of $\realcone(\Dd)$, we have $\sca{\ff}{\myvv}\geq 0$. As $\myvv\in \Lambda_\rr^B$, we also have $\sca{\ff}{\myvv}\equiv \sca{\ff}{\rr}\equiv p_\ff \bmod B$. The claim follows.
\end{subproof}

For a given vector $\myvv\in \Rr$, call
$\ff\in \Ff$ {\em{tight}} for $\myvv$ if $\sca{\myvv}{\ff}=p_\ff$. Importantly, \cref{cl:discreteSteps} implies the following: if $\ff$ is not tight for $\myvv$, then $\sca{\ff}{\myvv}\geq p_\ff+B$. Further, let
$$\Tight(\myvv)\coloneqq \{\ff\in \Ff~|~\ff\textrm{ is tight for }\myvv\}.$$
The next claim is the key step. Its proof relies on a carefully crafted application of the Deep-in-the-Cone Lemma (\cref{lem:cornetto}).

\begin{techclaim}\label{cl:tightImplies}
 Suppose $\uu\in \Zz$ and $\myvv\in \Rr$ are such that $\Tight(\uu)\supseteq \Tight(\myvv)$. Then $\myvv\in \Zz$ as well.
\end{techclaim}
\begin{subproof}
 Denote $\ell\coloneqq |\Ff|$ for brevity.
 Enumerate $\Ff$ as $\{\ff_1,\ff_2,\ldots,\ff_{\ell}\}$ so that
 $$\sca{\ff_1}{\myvv} \leq \sca{\ff_2}{\myvv}\leq \sca{\ff_3}{\myvv}\leq \ldots \leq \sca{\ff_{\ell}}{\myvv},$$
 and recall that all these values are nonnegative integers due to $\myvv$ being an integer vector in $\realcone(\Dd)$ and $\Ff$ being the dual representation of $\realcone(\Dd)$. Let $k\in \{0,1,\ldots,\ell\}$ be the largest index such that
 $$\sca{\ff_i}{\myvv}\leq B_i\qquad\textrm{for all }i\in \{1,2,\ldots,k\},$$
 and denote
 $$\Gg\coloneqq \{\ff_1,\ff_2,\ldots,\ff_k\}\qquad\textrm{and}\qquad \widehat{B}\coloneqq B_k.$$
 By the maximality of $k$ it follows that either $k=\ell$, or $k<\ell$ and $\sca{\ff_{k+1}}{\myvv}>B_{k+1}$. In any case, by the choice of the sequence $B_0,B_1,\ldots,B_\ell$ we have
 \begin{equation}\label{eq:wydra}
 \sca{\gg}{\myvv} \leq \widehat{B}\quad\textrm{for all }\gg\in \Gg\qquad\textrm{and}\qquad \sca{\ff}{\myvv}>\widehat{M}\widehat{B}\quad \textrm{for all }\ff\in \Ff\setminus \Gg.\end{equation}

 Since $\widehat{B}<B$, all vectors of $\Gg$ are tight for $\myvv$; see \cref{cl:discreteSteps}. As $\Tight(\uu)\supseteq \Tight(\myvv)$, these vectors are tight for $\uu$ as well. We conclude that
 \begin{equation}\label{eq:kapibara}
 \sca{\gg}{\uu}=\sca{\gg}{\myvv}=p_{\gg}\qquad \textrm{for all }\gg\in \Gg.
 \end{equation}

 Since $\uu\in \Zz$ by assumption, there is a collection of nonnegative integers $\mset{\lambda_\dd\colon \dd\in \Dd}$ such that
 $$\uu=\sum_{\dd\in \Dd} \lambda_\dd\cdot \dd.$$
 Define now a collection of nonnegative integers $\mset{\mu_\dd\colon \dd\in \Dd}$ as follows:
 \begin{itemize}[nosep]
  \item if $\dd\notin \Dd_\Gg$, then $\mu_\dd=\lambda_\dd$; and
  \item otherwise, if $\dd\in \Dd_\Gg$, then $\mu_\dd$ is the unique integer in $\{0,1,\ldots,K-1\}$ such that $\lambda_\dd\equiv \mu_\dd\bmod K$.
 \end{itemize}
 Consider now any $\ee\in \Dd\setminus \Dd_\Gg$. By the definition of $\Dd_\Gg$, there exists $\gg\in \Gg$ such that $\sca{\gg}{\ee}>0$, which by integrality entails $\sca{\gg}{\ee}\geq 1$. By \cref{obs:dualPositive},~\cref{eq:wydra}, and~\cref{eq:kapibara}, we have
 $$\widehat{B}\geq \sca{\gg}{\myvv}=\sca{\gg}{\uu}=\sum_{\dd\in \Dd} \lambda_\dd\cdot \sca{\gg}{\dd} \geq \lambda_{\ee}\cdot \sca{\gg}{\ee}\geq \lambda_\ee=\mu_\ee.$$
 Together with the straightforward inequality $\mu_\dd\leq K\leq \widehat{B}$ for all $\dd\in \Dd_\Gg$, we conclude that
 \begin{equation}\label{eq:nutria}
  \mu_\dd\leq \widehat{B}\qquad \textrm{for all }\dd\in \Dd.
 \end{equation}

 Let
 $$\uu'\coloneqq \sum_{\dd\in \Dd} \mu_\dd\cdot \dd.$$
 Clearly $\uu'\in \intcone(\Dd)$. Our goal is to show that \cref{lem:cornetto} can be applied to $\myvv-\uu'$.

 First, consider any $\ff\in \Ff\setminus \Gg$. Then by \cref{eq:wydra} and \cref{eq:nutria}, we have
 \begin{align}
\sca{\ff}{\myvv-\uu'} & =     \sca{\ff}{\myvv}-\sca{\ff}{\uu'} > \widehat{M}\widehat{B}-\sum_{\dd\in \Dd} \mu_\dd\cdot \sca{\ff}{\dd}\nonumber\\
                    & \geq  \widehat{M}\widehat{B}-|\Dd|\cdot \max_{\ff\in \Ff} \|\ff\|_1\cdot \max_{\dd\in \Dd} \|\dd\|_\infty\cdot \widehat{B} =  M\widehat{B} \geq M.\label{eq:mors}
 \end{align}
 Next, consider any $\gg\in \Gg$. Then
 \begin{equation}\label{eq:foka0}
\sca{\gg}{\myvv-\uu'}=\sum_{\dd\in \Dd} (\lambda_\dd-\mu_\dd)\cdot\sca{\gg}{\dd}=\sum_{\dd\in \Dd_\Gg} (\lambda_\dd-\mu_\dd)\cdot\sca{\gg}{\dd} + \sum_{\dd\in \Dd\setminus \Dd_\Gg} (\lambda_\dd-\mu_\dd)\cdot\sca{\gg}{\dd}.
 \end{equation}
 Note that for $\dd\in \Dd\setminus \Dd_\Gg$ we have $\lambda_\dd=\mu_\dd$, while for $\dd\in \Dd_\Gg$ we have $\sca{\gg}{\dd}=0$. So both summands on the right hand side of~\cref{eq:foka0} are equal to $0$, implying that
 \begin{equation}\label{eq:foka1}
\sca{\gg}{\myvv-\uu'}=0\qquad\textrm{for all }\gg\in \Gg.
 \end{equation}
 Observe that since $\Ff$ is the dual representation of $\realcone(\Dd)$, \cref{eq:mors} and \cref{eq:foka1} together imply that
 \begin{equation}\label{eq:morswin}
  \myvv-\uu'\in \realcone(\Dd).
 \end{equation}
 Then by combining \cref{eq:foka1}, \cref{eq:morswin}, and \cref{lem:facetRepresentation} we conclude that
 \begin{equation}\label{eq:foka}
  \myvv-\uu'\in \realcone(\Dd_\Gg).
 \end{equation}

 Finally, observe that for every $\dd\in \Dd$, $\lambda_\dd-\mu_\dd$ is an integer divisible by $K$. Therefore, $\uu-\uu'=K\cdot \ww$ for some integer vector $\ww$, implying that $\uu-\uu'\in \Lambda_{\zero}^{K}$.
 As $\uu,\myvv\in \Lambda_{\rr}^B$ and $K$ divides~$B$, we also have $\myvv-\uu\in \Lambda_{\zero}^{K}$. Combining these two observations yields
 \begin{equation}\label{eq:pizmak}
  \myvv-\uu'\in \Lambda_{\zero}^{K}.
 \end{equation}
 Noting that $\zero \in \lattice(\Dd_\Gg)\cap \Lambda_\zero^{K}$ and $K$ is divisible by $K_\Gg$, by applying \cref{lem:spaceCramer} to $\Dd_\Gg$ we infer that
 $$\spn(\Dd_\Gg)\cap \Lambda_{\zero}^{K}\subseteq \lattice(\Dd_\Gg).$$
 Combining this with \cref{eq:foka} and \cref{eq:pizmak} yields
 \begin{equation}\label{eq:manat}
\myvv-\uu'\in \lattice(\Dd_\Gg).
 \end{equation}

 Now, assertions \cref{eq:mors}, \cref{eq:foka}, and \cref{eq:manat} show that we can use \cref{lem:cornetto} to conclude that
 $$\myvv-\uu'\in \intcone(\Dd_\Gg).$$
 Since $\uu'=\sum_{\dd\in \Dd} \mu_\dd\cdot \dd$ and $\mset{\mu_\dd\colon \dd\in \Dd}$ are nonnegative integers, we have $\uu'\in \intcone(\Dd_\Gg)$ as well. So as a sum of two vectors from $\intcone(\Dd_\Gg)$, $\myvv$ also belongs to $\intcone(\Dd_\Gg)$, and we are done.
\end{subproof}

For every $\Gg\subseteq \Ff$, let
$$\Rr_\Gg\coloneqq \{\myvv\in \Rr~|~\Tight(\myvv)=\Gg\}.$$
In other words, $\Rr_\Gg$ comprises all vectors in $\Rr$ for which $\Gg$ is exactly the set of tight vectors in $\Ff$. Thus, $\{\Rr_\Gg\colon \Gg\subseteq \Ff\}$ is a partition of $\Rr$.

Let $\Ss$ be the family of subsets of $\Ff$ consisting of all $\Gg\subseteq \Ff$ such that $\Rr_\Gg\cap \Zz\neq \emptyset$. Further, let $\Ss^{\downarrow}$ be the downward closure of $\Ss$: a set $\Gg\subseteq \Ff$ belongs to $\Ss^\downarrow$ if and only if there exists $\Gg'\supseteq \Gg$ such that $\Gg'\in \Ss$.
The following statement is an immediate corollary of \cref{cl:tightImplies}.

\begin{techclaim}\label{cl:allornothing}
Suppose $\Gg\subseteq \Ff$. If $\Gg\in \Ss^{\downarrow}$ then $\Rr_\Gg\subseteq \Zz$, and if $\Gg\notin \Ss^{\downarrow}$ then $\Rr_\Gg\cap \Zz=\emptyset$. Consequently,
$$\Zz=\bigcup_{\Gg\in \Ss^\downarrow} \Rr_\Gg.$$
\end{techclaim}

Next, we show that the characterization of \cref{cl:allornothing} can be expressed through linear inequalities.

\begin{techclaim}\label{cl:constructionQ}
 For every $\myvv\in \Lambda_\rr^B$, the following conditions are equivalent:
 \begin{enumerate}[label=(\arabic*),ref=(\arabic*),nosep]
  \item\label{c:inZ} $\myvv\in \intcone(\Dd)$; and
  \item\label{c:inQ} $\myvv\in \realcone(\Dd)$ and
  $$\sum_{\gg\in \Gg} \sca{\gg}{\myvv} \geq 1+\sum_{\gg\in \Gg} p_\gg\qquad\textrm{for all }\Gg\subseteq \Ff\textrm{ such that }\Gg\notin \Ss^\downarrow.$$
 \end{enumerate}
\end{techclaim}
\begin{subproof}
We first prove implication \ref{c:inZ}$\Rightarrow$\ref{c:inQ}. Since $\myvv\in \Lambda_\rr^B\cap \intcone(\Dd)=\Zz$, by \cref{cl:allornothing} we have that $\myvv\in \Rr_\Gg$ for some $\Gg\in \Ss^\downarrow$. In particular, for every $\Gg'\subseteq \Ff$ with $\Gg'\notin \Ss^\downarrow$, $\Gg'$ is not a subset of $\Gg$, implying that there exists some $\hh\in \Gg'\setminus \Gg$. As $\Gg=\Tight(\myvv)$, $\hh$ is not tight for $\myvv$, hence $\sca{\hh}{\myvv}\geq p_\hh+B$ by \cref{cl:discreteSteps}. Hence, by  \cref{cl:discreteSteps} again,
$$\sum_{\gg\in \Gg'} \sca{\gg}{\myvv}\geq B+\sum_{\gg\in \Gg'} p_\gg\geq 1+\sum_{\gg\in \Gg'} p_\gg;$$
and this holds for each $\Gg'\subseteq \Ff$ with $\Gg'\notin \Ss^{\downarrow}$.
As $\myvv\in \intcone(\Dd)$ entails $\myvv\in \realcone(\Dd)$, this proves~\ref{c:inQ}.

We now move to the implication \ref{c:inQ}$\Rightarrow$\ref{c:inZ}. Since $\myvv\in  \Lambda_\rr^B\cap \realcone(\Dd)=\Rr$, by \cref{cl:allornothing} it suffices to show that the unique $\Gg\subseteq \Ff$ for which $\myvv\in \Rr_\Gg$ satisfies $\Gg\in \Ss^\downarrow$. Note $\myvv\in \Rr_\Gg$ is equivalent to $\Gg=\Tight(\myvv)$. Therefore $\sca{\gg}{\myvv}=p_\gg$ for all $\gg\in \Gg$, implying that
$$\sum_{\gg\in \Gg} \sca{\gg}{\myvv} = \sum_{\gg\in \Gg} p_\gg.$$
Hence we necessarily must have $\Gg\in \Ss^\downarrow$, for otherwise \ref{c:inQ} would not be satisfied.
\end{subproof}

We may now define $\Qq$ to be the set of the following pairs:
\begin{itemize}[nosep]
 \item $(\ff,0)$ for all $\ff\in \Ff$; and
 \item $\left(\sum_{\gg\in \Gg} \gg,1+\sum_{\gg\in \Gg} p_\gg\right)$ for all $\Gg\subseteq \Ff$ with $\Gg\notin \Ss^\downarrow.$
\end{itemize}
As $\Ff$ is a dual representation of $\realcone(\Dd)$, we have $\realcone(\Dd)=\{\myvv\in \R^\tv~|~\sca{\ff}{\myvv}\geq 0\textrm{ for all }\ff\in \Ff\}$. Hence, \cref{cl:constructionQ} directly implies that
$$\Lambda_\rr^B\cap \intcone(\Dd)=\Lambda_\rr^B\cap \left\{\myvv\in \R^\tv~|~\sca{\qq}{\myvv}\geq a\textrm{ for all }(\qq,a)\in \Qq\right\},$$
as required.

\medskip

It remains to argue that the set $\Qq$ can be computed given $\Dd$ and $\rr$. For this, it suffices to distinguish which sets $\Gg\subseteq \Ff$ belong to $\Ss$ and which not, because based on this knowledge we may compute $\Ss^{\downarrow}$ and then construct $\Qq$ right from the definition. (Recall here that by \cref{thm:weyl}, $\Ff$ can be computed from~$\Dd$.) Testing whether given $\Gg\subseteq \Ff$ belongs to $\Ss$ boils down to verifying whether there exists $\myvv\in \Z^\tv$ such that
\begin{itemize}[nosep]
 \item $\myvv\in \Lambda_\rr^B$, or equivalently, $\myvv=B\cdot \ww+\rr$ for some $\ww\in \Z^\tv$;
 \item $\sca{\gg}{\myvv}=p_\gg$ for all $\gg\in \Gg$; and
 \item $\sca{\gg}{\myvv}\geq p_\gg+1$ for all $\gg\in \Ff\setminus \Gg$; and
 \item $\myvv\in \lattice(\Dd)$, or equivalently, there exist nonnegative integer coefficients $\mset{\lambda_\dd\colon \dd\in \Dd}$ such that $\myvv=\sum_{\dd\in \Dd} \lambda_\dd\cdot \dd$.
\end{itemize}
These conditions form an integer program with $2|\tv|+|\Dd|$ variables: $|\tv|$ variables for $\myvv$, $|\tv|$ variables for $\ww$, and $|\Dd|$ variables for the coefficients $\mset{\lambda_\dd\colon \dd\in \Dd}$. Hence we may just check the feasibility of this program using any of the standard algorithms for integer programming, e.g., the algorithm of Kannan~\cite{kannan1987minkowski}.
\end{proof}

We remark that while the statement of \cref{thm:polyhedron} does not specify any concrete upper bound on the value of $B$, it is not hard to trace all the estimates used throughout the reasoning to see that $B$ is bounded by an elementary function (that is, a constant-height tower of exponents) of the relevant parameters $|\tv|$, $|\Dd|$, and $\max_{\dd\in \Dd}\|\Dd\|_\infty$. We did not attempt to optimize the value of $B$ in our proof, hence finding tighter estimates is left to future work.

\subsection{Algorithm for two-stage stochastic integer programming}

With all the tools prepared, we are ready to give a proof of \cref{thm:2stage-main}.

\begin{proof}[Proof of Theorem~\ref{thm:2stage-main}]
 For each matrix $D_i$, let $\Dd_i$ be the set of columns of $D_i$. Since every member of $\Dd_i$ is a vector in $\Z^\tv$ with all entries of absolute value at most $\Delta$, and there are $(2\Delta+1)^{|\tv|}$ different such vectors, we have $|\Dd_i|\leq (2\Delta+1)^{|\tv|}$. In particular, there are at most $2^{(2\Delta+1)^{|\tv|}}$ distinct sets~$\Dd_i$.

 Now, the feasibility of the program $P$ is equivalent to the following assertion: there exist $\uu\in \Zn^\xx$ such that
 \begin{equation}\label{eq:szynszyla}
   \bb_i-A_i\uu\in \intcone(\Dd_i)\qquad\textrm{for all }i\in I.
 \end{equation}
 For each $i\in I$ apply \cref{thm:polyhedron} to $\Dd_i$, yielding a positive integer $B_i$, computable from $\Dd_i$. Let $B$ be the least common multiple of all integers $B_i$ for $i\in I$. Since the number of distinct sets $\Dd_i$ is bounded by~$2^{(2\Delta+1)^{|\tv|}}$, it follows that $B$ is bounded by a computable function of $\Delta$ and $|\tv|$.

 Towards verification of assertion~\cref{eq:szynszyla}, the algorithm guesses, by iterating through all possibilities, the vector $\rr\in \{0,1,\ldots,B-1\}^{\xx}$ such that $\uu[x]\equiv \rr[x]\bmod B$ for each $x\in \xx$; equivalently $\uu\in \Lambda_\rr^B$. This uniquely defines, for each $i\in I$, a vector $\rr_i\in \{0,1,\ldots,B_i-1\}^{\tv}$ such that
 $$(\bb_i-A_i\uu)[t]\equiv \rr_i[t]\mod B_i\qquad\textrm{for all }t\in \tv.$$
 Or equivalently, $\bb_i-A_i\uu\in \Lambda^{B_i}_{\rr_i}$.

 Apply the algorithm of \cref{thm:polyhedron} to $\Dd_i$ and $\rr_i$, yielding a set of pairs $\Qq_i\subseteq \Z^\tv\times \Z$ such that
 $$\Lambda^{B_i}_{\rr_i}\cap \intcone(\Dd_i) = \Lambda^{B_i}_{\rr_i}\cap \{\myvv\in \R^\tv~|~\sca{\qq}{\myvv}\geq a\textrm{ for all }(\qq,a)\in \Qq_i\}.$$
 Hence, under the supposition that indeed $\uu\in \Lambda_\rr^B$, assertion~\cref{eq:szynszyla} boils down to verifying the existence of $\uu\in \Z^{\xx}$ satisfying the following:
 \begin{itemize}[nosep]
  \item $\uu[x]\geq 0$ for all $x\in \xx$;
  \item $\uu\in \Lambda_\rr^B$, or equivalently, $\uu=B\cdot \ww+\rr$ for some $\ww\in \Z^\xx$;
  \item for each $i\in I$ and each $(\qq,a)\in \Qq_i$, we have
  $$\sca{\qq}{\bb_i-A_i\uu}\geq a.$$
 \end{itemize}
 The conditions above form an integer program with $2|\xx|$ variables ($|\xx|$ variables for $\uu$ and $|\xx|$ variables for $\ww$) whose total bitsize is bounded by $g(\Delta,|\tv|)\cdot \|P\|$ for some computable function $g$. Hence, we can solve this program (and thus verify the existence of a suitable $\uu$) in time $h(\Delta,|\xx|,|\tv|)\cdot \|P\|$ for a computable $h$ using, for instance, the algorithm of Kannan~\cite{kannan1987minkowski}. Iterating through all $\rr\in \{0,1,\ldots,B-1\}^\xx$ adds only a multiplicative factor that depends in a computable manner on $\Delta,|\xx|,|\tv|$, yielding in total a time complexity bound of $f(\Delta,|\xx|,|\tv|)\cdot \|P\|$ for a computable function $f$, as claimed.
\end{proof}

\section{\texorpdfstring{\nfold{}}{n-fold} integer programming}

Our main result for \nfold{} integer programming is presented in the following statement, which is just \cref{thm:nfold-intro} with adjusted notation.

\begin{theorem}\label{thm:nfold-main}
An instance $P=(C,\av,\mset{D_i,\bb_i,\cc_i\colon i\in I})$ of {\textsc{Uniform}} \nfoldOpt{} can be solved in time $f(\Delta,|\yy|,|\tv|)\cdot \|P\|^{\Oh(1)}$ for some computable function $f$, where $\Delta=\max_{i\in I}\|D_i\|_\infty$.
\end{theorem}

Again, importantly, we do not impose any upper bound on the absolute values of the entries of the matrix $C$. Note also that as mentioned in \cref{sec:intro}, the considered parameters {\em{do not}} include $|\sv|$, the number of linking constraints, but {\em{do}} include $|\yy|$, the number of local variables in each block.

The remainder of this section is devoted to the proof of \cref{thm:nfold-main}.

\subsection{Graver bases}\label{sec:Gravers}

One ingredient that we will repeatedly use in our proofs is the notion of the Graver basis of a matrix, which consists of $\cfleq$-minimal integral elements of the kernel. Formally, the {\em{integer kernel}} of a matrix $D \in \matrices{\Z}{\yy}{\tv}$, denoted $\ker^\Z(D)$, is the set of all integer vectors~$\uu \in \Z^{\tv}$ such that $D\uu = \mathbf{0}$. The {\em{Graver basis}} of $D$ is the set $\Graver(D)$ of all $\cfleq$-minimal non-zero elements of $\ker^\Z(D)$. In other words, a vector $\uu\in \ker^\Z(D)\setminus \{\mathbf{0}\}$ belongs to $\Graver(D)$ if there is no $\uu'\cfleq \uu$, $\uu'\notin \{\mathbf{0},\uu\}$, such that $\uu'\in \ker^\Z(D)$ as~well.

Note that since $\Graver(D)$ is an antichain in the $\cfleq$ order, which is a well quasi-order on $\Z^\tv$ by Dickson's Lemma, $\Graver(D)$ is always finite. There are actually many known bounds on the norms of the elements of the Graver basis under different assumptions about the matrix. We will use the following one tailored to matrices with few rows.

\begin{lemma}[Lemma~2 of \cite{EisenbrandHK18}]\label{lem:GraverBound}
Let $D \in \matrices{\Z}{\yy}{\tv}$ be an integer matrix, and let $\Delta=\|D\|_\infty$. Then for every $\gg\in \Graver(D)$, it holds that
\begin{align*}
    \|\gg\|_1 \leq (2|\tv|\Delta+1)^{|\tv|}
\end{align*}
\end{lemma}

We will also use the fact that every solution to an integer program can be decomposed into some solution of bounded norm and a multiset consisting of elements of the Graver basis. This fact was observed by Pottier in~\cite[Corollary~1]{Pottier91}, we recall his proof for completeness.

\begin{lemma}\label{lem:GraverDecomp}
Let $D \in \matrices{\Z}{\yy}{\tv}$ be an integer matrix and $\bb \in \Zn^{\tv}$ be an integer vector. Then for every $\ww\in \Zn^{\yy}$ such that $D\ww=\bb$, there exists a vector $\whww\in \Zn^{\yy}$ and a multiset $\Gg$ consisting of nonnegative elements of $\Graver(D)$ such that $$D\whww=\bb, \qquad\|\whww\|\leq (2|\tv|(\|D\|_\infty+\|\bb\|_\infty)+1)^{|\tv|},\qquad \textrm{and}\qquad \ww=\whww+\sum_{\gg\in \Gg} \gg.$$
\end{lemma}
\begin{proof}
Let $z$ be a fresh variable that does not belong to $\yy$ and let $\zz\coloneqq \yy\cup \{z\}$. Further, let $D'\in \matrices{\Z}{\zz}{\tv}$ be the matrix obtained from $D$ by adding column $D'[z]=-\bb$, and $\ww'\in \Zn^{\zz}$ be the vector obtained from $\ww$ by adding the entry $\ww'[z]=1$. Observe that $D\ww=\bb$ entails $D'\ww'=\mathbf{0}$, which means that $\ww'\in \ker^\Z(D')$. Every element of the integer kernel can be decomposed into a sign-compatible sum of elements of the Graver basis; see e.g.~\cite[Lemma~3.2.3]{JesusBook}. Therefore, there is a multiset $\Gg'$ consisting of nonnegative elements of $\Graver(D')$ such that $\sum_{\gg'\in \Gg'}\gg'=\ww'$.
Due to the sign-compatibility, exactly one element $\whww' \in \Gg'$ has $\whww[z]=1$, and all the other elements $\gg'\in \Gg'\setminus \{\whww'\}$ satisfy $\gg'[z]=0$. We can now obtain $\whww$ and $\Gg$ from $\whww'$ and $\Gg'\setminus \{\whww'\}$, respectively, by stripping every vector from the coordinate corresponding to variable $z$ (formally, restricting the domain to $\yy$). As $\gg'[z]=0$ for all $\gg'\in \Gg'\setminus \{\whww'\}$, it follows that all elements of $\Gg$ belong to $\Graver(D)$. Finally, from \cref{lem:GraverBound} applied to $D'$ we conclude that $\|\whww\|_\infty\leq \|\whww'\|_\infty\leq (2|\tv|(\|D\|_\infty+\|\bb\|_\infty)+1)^{|\tv|}$.
\end{proof}

\subsection{Decomposing bricks}

As explained in \cref{sec:overview}, the key idea in the proof of \cref{thm:nfold-main} is to decompose right-hand sides of the program (that is, vectors $\bb_i$) into smaller and smaller vectors while preserving the optimum value of the program. The next lemma is the key observation that provides a single step of the decomposition.

\begin{lemma}[Brick Decomposition Lemma]\label{lem:Decomp}
There exists a function $g(k,\Delta)\in 2^{(k\Delta)^{\Oh(k)}}$ such that the following holds.
Let $D\in \matrices{\Z}{\yy}{\tv}$ be a matrix and $\bb\in \Z^{\tv}$ be a vector such that $\|\bb\|_\infty>g(k,\Delta)$,
 where $k=|\tv|$ and $\Delta=\|D\|_\infty$. Then there are non-zero vectors $\bb',\bb''\in \Z^\tv$ such that the following conditions hold:
\begin{itemize}[nosep]
    \item $\bb',\bb''\cfleq \bb$ and $\bb=\bb'+\bb''$; and
    \item for every $\myvv\in \Zn^{\yy}$ satisfying $D\myvv=\bb$, there exist $\myvv',\myvv''\in \Zn^{\yy}$ such that
    $$\myvv=\myvv'+\myvv'',\qquad D\myvv'=\bb',\qquad\textrm{and}\qquad  D\myvv''=\bb''.$$
\end{itemize}
\end{lemma}

Throughout this section, for a multiset $V$ of vectors or numbers, we denote by $\sum V$ the sum of the elements in $V$.

The crucial ingredient to prove \cref{lem:Decomp} is the following result of Klein~\cite{DBLP:journals/mp/Klein22}. We use the variant from Cslovjecsek et al.~\cite{CslovjecsekEPVW21}, which compared to the original version of Klein~\cite{DBLP:journals/mp/Klein22} has better bounds and applies to vectors with possibly negative coefficients.

\begin{lemma}[Klein Lemma, Theorem~4.1 of~\cite{CslovjecsekEPVW21} for $\epsilon=1$ and with adjusted notation]\label{lem:Klein}
Let $\tv$ be a tuple of variables with $|\tv|=k$ and let $T_1,\ldots,T_n$ be non-empty multisets of vectors in $\Z^{\tv}$ such that for all $i\in \{1,\ldots,n\}$ and $\uu\in T_i$, we have $\|\uu\|_\infty\leq \Delta$, and the sum of all elements in each multiset is the same:
\begin{align*}
    \sum T_1 =\sum T_2 = \ldots = \sum T_n.
\end{align*}
Then there exist non-empty submultisets $S_1 \subseteq T_1,S_2 \subseteq T_2,\ldots,S_n \subseteq T_n$, each of size bounded  by $2^{\Oh(k\Delta)^k}$, such that
\begin{align*}
    \sum S_1 =\sum S_2 = \ldots = \sum S_n.
\end{align*}
\end{lemma}

We also need the following lemma about partitioning a multiset of vectors into  small submultisets with sign-compatible sums.

\begin{lemma}\label{lem:bundling}
    Let $U$ be a multiset of vectors in $\Z^\tv$ with $\|\uu\|_\infty\leq \Delta$ for all $\uu\in U$, where $\tv$ is a tuple of variables with $|\tv|=k$. Further, let $\bb=\sum U$. Then one can partition $U$ into a collection of non-empty submultisets~$\mset{U_i\colon i\in I}$ so that for every $i\in I$, we have $|U_i|\leq (2k\Delta+1)^{k}$ and $\sum U_i \cfleq \bb$.
\end{lemma}
\begin{proof}
    Let $W$ be a multiset of vectors in $\Z^{\tv}$ such that
    \begin{itemize}[nosep]
        \item $\ww\cfleq -\bb$ and $\|\ww\|_\infty\leq 1$ for each $\ww\in W$, and
        \item $\sum \ww = -\bb$.
    \end{itemize}
    Let $\xx$ and $\yy$ be tuples of variables enumerating the elements of $U$ and $W$, respectively, so that we can set up a matrix $A\in \Z^{(\xx\cup \yy)\times \tv}$ with the elements of $U$ as the columns corresponding to $\xx$ and the elements of $W$ as the columns corresponding to $\yy$. Thus, we have
    \begin{align*}
        A\mathbf{1}=\zero.
    \end{align*}
    By \cref{lem:GraverBound}, there exists $\gg\in \Graver(A)$ such that
    \begin{align*}
        \|\gg\|_1\leq (2k\Delta+1)^{k}\qquad\textrm{and}\qquad \gg\cfleq \mathbf{1}.
    \end{align*}
    Let $U_1$ be the multisets of those elements of $U$ that correspond to the variables $x\in \xx$ with $\gg(x)=1$. Similarly, let $W_1$ be the multiset of those elements of $W$ that correspond to the variables $y\in \yy$ with $\gg(y)=1$. Thus, we have $|U_1|+|W_1|=\|\gg\|_1\leq (2k\Delta+1)^{k}$. Moreover, as $A\gg=\zero$, $\sum W=-\bb$, and $\ww\cfleq -\bb$ for each $\ww\in W_1$, we have
    \[\sum U_1=-\sum W_1\cfleq \bb.\]
    Therefore, we can set $U_1$ to be the first multiset in the sought collection. It now remains to apply the same reasoning inductively to the multiset $U\setminus U_1$ in order to extract the next elements of the collection, until the multiset under consideration becomes empty.
\end{proof}

We are now ready to prove \cref{lem:Decomp}.

\begin{proof}[Proof of Lemma~\ref{lem:Decomp}]
Let $\Xi\in 2^{\Oh(k\Delta)^k}$ be the bound provided by \cref{lem:Klein} for parameters $k$ and $\Delta$.

Call a vector $\myvv\in \Zn^\yy$ a {\em{solution}} if $D\myvv=\bb$. Further, a solution $\myvv$ is {\em{minimal}} if it is conformally minimal among the solutions: there is no solution $\myvv'$ such that $\myvv'\neq \myvv$ and $\myvv'\cfleq \myvv$.
Let $V$ be the set of all minimal solutions. Note that $V$ is an antichain in $\cfleq$, so as $\cfleq$ is a well quasi-order on $\Zn^\yy$ by Dickson's Lemma, it follows that $V$ is finite. If $V$ is empty, then there are no solutions and there is noting to prove, so assume otherwise.

For each $\myvv\in V$, construct a multiset of vectors $T_\myvv$ as follows: for each $y\in \yy$, include $\myvv[y]$ copies of the column $D[y]$ (i.e., the column of $D$ corresponding to variable $y$) in $T_\myvv$. Thus, we obtain a finite collection of multisets $\mset{T_\myvv\colon \myvv\in V}$. By construction, we have
$$\sum T_\myvv = D\myvv=\bb\qquad\textrm{for all }\myvv\in V.$$
Note that all vectors in all multisets $T_\myvv$ have all entries bounded in absolute value by $\Delta$. In the sequel, we will use the following claim a few times.

\begin{techclaim}\label{cl:minimal}
    Suppose $\myvv\in V$ and $A\subseteq T_\myvv$ is a submultiset such that $\sum A=\mathbf{0}$. Then $A$ is empty.
\end{techclaim}
\begin{subproof}
    Let $\ww\in \Zn^\yy$ be the vector of multiplicities in which the columns of $D$ appear in $A$: for $y\in \yy$, $\ww[y]$ is the number of times $D[y]$ appears in $A$. Clearly, $\ww\cfleq \myvv$ and $\sum A=\mathbf{0}$ implies that $D\ww=\mathbf{0}$. Then $D(\myvv-\ww)=\bb$ and $\myvv-\ww\cfleq \myvv$ as well. This is a contradiction with the minimality of $\myvv$ unless $\ww=\mathbf{0}$, or equivalently, $A$ is empty.
\end{subproof}

As $V$ is finite, we may apply \cref{lem:Klein} to multisets $\mset{T_\myvv\colon \myvv\in V}$. In this way, we obtain submultisets $\mset{S_\myvv\subseteq T_\myvv\colon \myvv\in V}$, each of size at most $\Xi$,
such that
\begin{align*}
    \sum S_\myvv = \sum S_{\myvv'}\qquad \textrm{for all }\myvv,\myvv'\in V.
\end{align*}

Consider multisets $T'_\myvv\coloneqq T_\myvv\setminus S_\myvv$ for $\myvv\in V$. Observe that the multisets $T'_\myvv$ have again the same sums. Moreover, if some $T'_\myvv$ is empty, then $\sum T'_{\myvv'}=0$ for every $\myvv'\in V$, which by \cref{cl:minimal} implies that $T'_{\myvv'}$ is empty as well. It follows that either all multisets $\mset{T'_\myvv\colon \myvv\in V}$ are empty, or all of them are non-empty.

If all multisets $\mset{T'_\myvv\colon \myvv\in V}$ are non-empty, then we may apply the same argument to them again, and thus extract suitable non-empty submultisets $\mset{S'_\myvv\subseteq T'_\myvv\colon \myvv\in V}$ with the same sum. By performing this reasoning iteratively until all multisets become empty (which occurs simultaneously due to \cref{cl:minimal}), we find a collection of partitions
$$\mset{\mset{S_\myvv^i\colon i\in I}\colon \myvv\in V}$$
for some index set $I$ such that:
\begin{itemize}[nosep]
    \item $\mset{S_\myvv^i\colon i\in I}$ is a partition of $T_\myvv$ for each $\myvv\in V$;
    \item $S_\myvv^i$ is non-empty and $|S_\myvv^i|\leq \Xi$ for all $i\in I$ and $\myvv\in V$; and
    \item $\sum S^i_\myvv=\sum S^i_{\myvv'}$ for all $i\in I$ and $\myvv,\myvv'\in V$.
\end{itemize}
For $i\in I$, let $\pp_i$ be the common sum of multisets $S^i_\myvv$ for $\myvv\in V$, that is,
$$\pp_i= \sum S^i_\myvv\qquad\textrm{for all }\myvv\in V.$$
Note that $\sum_{i\in I}\pp_i = \bb$ and $\|\pp_i\|_\infty\leq \Delta\Xi\eqqcolon \Delta'$ for all $i\in I$.

We now apply \cref{lem:bundling} to the collection of vectors $\mset{\pp_i\colon i\in I}$. This yields a partition $\mset{I_j\colon j\in J}$ of $I$ with some index set $J$ such that for each $j\in J$, we have
$$|I_j|\leq (2k\Delta'+1)^{k}\leq  2^{(k\Delta)^{\Oh(k)}}\qquad\textrm{and}\qquad \sum_{i\in I_j} \pp_i\cfleq \bb.$$
By setting $g(k,\Delta)\in 2^{(k\Delta)^{\Oh(k)}}$ large enough, we may guarantee that $|I_j|\cdot \Delta\leq g(k,\Delta)$ for all $j\in J$. Since we assumed that $\|b\|_\infty>g(k,d)$, we conclude that $\|b\|_\infty>\sum_{i\in I_j} \pp_i$ for all $j\in J$, and consequently it must be the case that $|J|\geq 2$.

Let $\{J',J''\}$ be any partition of $J$ with both $J'$ and $J''$ being non-empty. Define
$$I'\coloneqq \bigcup_{j\in J'} I_j\qquad \textrm{and}\qquad I''\coloneqq \bigcup_{j\in J''} I_j.$$
Further, we set
$$\bb'\coloneqq \sum_{i\in I'} \pp_i\qquad \textrm{and}\qquad\bb''\coloneqq \sum_{i\in I''} \pp_i.$$
As $\sum_{i\in I} \pp_i=\bb$ and $\pp_i\cfleq \bb$ for all $i\in I$, it follows that $\bb'+\bb''=\bb$ and $\bb',\bb''\cfleq \bb$. Further, from \cref{cl:minimal} we have that both $\bb'$ and $\bb''$ are non-zero.

Consider any $\myvv\in V$. Let
$$S'_\myvv\coloneqq \bigcup_{i\in I'} S^i_\myvv\qquad \textrm{and}\qquad S''_\myvv\coloneqq \bigcup_{i\in I''} S^i_\myvv .$$
Note that
\begin{equation}\label{eq:zaba}
\sum S'_\myvv = \bb'\qquad \textrm{and}\qquad \sum S''_\myvv = \bb''.
\end{equation}
Further, let $\myvv'\in \Zn^{\yy}$ be the vector of multiplicities in which the columns of $D$ appear in $S'_\myvv$: for $y\in \yy$, $\myvv'[y]$ is the number of times $D[y]$ appears in $S'_\myvv$. Define $\myvv''$ analogously for $S''_\myvv$. Since $\{S'_\myvv,S''_\myvv\}$ is a partition of $T_\myvv$, we have $\myvv'+\myvv''=\myvv$ and $\myvv',\myvv''\cfleq \myvv$. Moreover, from~\cref{eq:zaba} it follows that
\begin{equation*}
D\myvv' = \bb'\qquad \textrm{and}\qquad D\myvv'' = \bb''.
\end{equation*}

\newcommand{\whvv}{\widehat{\myvv}}

The above reasoning already settles the second condition from the lemma statement for every minimal solution. So consider now any solution $\myvv\in \Zn^\yy$. Then there exists a minimal solution $\whvv\in V$ such that $\whvv\cfleq \myvv$. As both $\myvv$ and $\whvv$ are solutions, we have $D(\myvv-\whvv)=\mathbf{0}$. It follows that we may simply take
$$\myvv'\coloneqq \whvv'+(\myvv-\whvv)\qquad \textrm{and}\qquad \myvv''\coloneqq \whvv''.$$
This concludes the proof.
\end{proof}

From now on, we adopt the function $g(\Delta,k)$ provided by \cref{lem:Decomp} in the notation.

Now that \cref{lem:Decomp} is established, it is tempting to apply it iteratively: first break $\bb$ into $\bb'$ and $\bb''$, then break $\bb'$ into two even smaller vectors, and so on. To facilitate the discussion about such decompositions, we introduce the following definition.

\begin{definition}
    Let $D\in \matrices{\Z}{\yy}{\tv}$ be a matrix and $\bb\in \Z^\tv$ be a vector. A {\em{faithful decomposition}} of~$\bb$ with respect to $D$ is a collection $\mset{\bb_i\colon i\in I}$ of non-zero vectors in $\Z^\tv$ satisfying the following conditions:
\begin{itemize}[nosep]
    \item $\bb_i\cfleq \bb$ for each $i\in I$ and $\bb=\sum_{i\in I} \bb_i$; and
    \item for every $\myvv\in \Zn^{\yy}$ satisfying $D\myvv=\bb$, there exist a collection of vectors $\mset{\myvv_i\colon i\in I}$ in $\Zn^{\yy}$ such that
    $$\myvv=\sum_{i\in I} \myvv_i\qquad \textrm{and} \qquad D\myvv_i=\bb_i\quad\textrm{for each }i\in I.$$
\end{itemize}
The {\em{order}} of a faithful decomposition $\mset{\bb_i\colon i\in I}$ is $\max_{i\in I} \|\bb_i\|_\infty$.
\end{definition}

By applying \cref{lem:Decomp} iteratively, we get the following:

\begin{lemma}\label{lem:ConformalDecomp}
For every matrix $D\in \matrices{\Z}{\yy}{\tv}$ and a non-zero vector $\bb\in \Z^{\tv}$, there exists a faithful decomposition of $\bb$ with respect to $D$ of order at most $g(\Delta,k)$, where $k=|\tv|$ and $\Delta=\|D\|_\infty$.
\end{lemma}
\begin{proof}
Start with a faithful decomposition $\Bb$ consisting only of $\bb$. Then, as long as $\Bb$ contains a vector $\pp$ with $\|\pp\|>g(k,\Delta)$, apply \cref{lem:Decomp} to $\pp$ yielding suitable vectors $\pp'$ and $\pp''$, and replace $\pp$ with $\pp'$ and $\pp''$ in $\Bb$. It is straightforward to verify that throughout this procedure $\Bb$ remains a faithful decomposition. Moreover, the size of $\Bb$ increases in each step of the procedure, while every faithful decomposition of $\bb$ has cardinality bounded by $\|\bb\|_1$. Therefore, the procedure must eventually end, yielding a faithful decomposition of order at most $g(k,\Delta)$.
\end{proof}

Observe that in the worst case, a faithful decomposition $\Bb$ provided by \cref{lem:ConformalDecomp} can consist of as many as $\Omega(\|\bb\|_1)$ vectors, which in our case is exponential in the size of the bit encoding of $\bb$. However, as all vectors in $\Bb$ have $\ell_\infty$-norm bounded in terms of $k$ and $\Delta$, the total number of distinct vectors is bounded by a function of $k$ and $\Delta$. Therefore, a faithful decomposition $\Bb$ provided by \cref{lem:ConformalDecomp} can be encoded compactly using $g(k,\Delta)\cdot \log \|\bb\|_\infty$ bits, by storing each vector present in $\Bb$ together with its multiplicity in $\Bb$, encoded in binary. In all further algorithmic discussions we will assume that this encoding scheme for faithful decompositions.

Our next goal is to prove that the faithful decomposition of \cref{lem:ConformalDecomp} can be computed algorithmically, in fixed-parameter time when parameterized by $|\yy|$, $k$, and $\Delta$, assuming $\bb$ is given on input in binary. The idea is to iteratively extract a large fraction of the elements of the decomposition, so that the whole process finishes after a number of steps that is logarithmic in $\|\bb\|_\infty$. Each extraction step will be executed using a fixed-parameter algorithm for deciding Presburger Arithmetic. Let us give a brief introduction.

Presburger Arithmetic is the first-order theory of the structure $\langle \Zn,+,0,1,2,\ldots \rangle$, that is, of nonnegative integers equipped with addition (treated as a binary function) and all elements of the universe as constants. More precisely, a {\em{term}} is an arithmetic expression using the binary function $+$, variables (meant to be evaluated to nonnegative integers), and constants (nonnegative integers). Formulas considered in Presburger Arithmetic are constructed from atomic formulas --- equalities of terms --- using the standard syntax of first-order logic, including standard boolean connectives, negation, and both existential and universal quantification (over nonnegative integers). The semantics is as expected. Note that while comparison is not directly present in the signature, it can be easily emulated by existentially quantifying a slack variable: for two variables $x,y$, the assertion $x\leq y$ is equivalent to $\exists_z\ x+z=y$. As usual, $\varphi(\xx)$ denotes a formula with a tuple of free variables $\xx$, while a {\em{sentence}} is a formula without free variables. The {\em{length}} of a formula $\varphi(\xx)$, denoted~$\|\varphi\|$, is defined in a standard way by structural induction on the formula and the terms contained within. Here, all constants are deemed to have length~$1$.

Presburger~\cite{Presburger29} famously proved that Presburger Arithmetic is decidable: there is an algorithm that given a sentence $\varphi$ of first-order logic over $\langle \Zn,+,0,1,2,\ldots \rangle$, decides whether $\varphi$ is true. As observed by Kouteck\'y and Talmon~\cite{KouteckyT21}, known algorithms for deciding Presburger Arithmetic can be understood as fixed-parameter algorithms in the following sense.

\begin{theorem}[follows from {\cite[Theorem~1]{KouteckyT21}}]\label{thm:PA}
Given a first-order sentence $\varphi$ over the structure $\langle \Zn,+,0,1,2,\ldots \rangle$, with constants encoded in binary, one can decide whether $\varphi$ is true in time $f(\|\varphi\|)\cdot (\log \Delta)^{\Oh(1)}$, where $\Delta\geq 3$ is an upper bound on all the constants present in $\varphi$ and $f$ is a computable function.
\end{theorem}

We remark that Kouteck\'y and Talmon discuss a more general setting where formulas are written in a more concise form, including allowing a direct usage of modular arithmetics, and one considers minimization of convex functions over tuples of variables satisfying constraints expressible in Presburger Arithmetic. Nevertheless, \cref{thm:PA}, as stated above, follows readily from \cite[Theorem~1]{KouteckyT21}. For a proof of \cite[Theorem~1]{KouteckyT21}, see~\cite[Theorem~2.2]{KouteckyT20}.


We now use \cref{thm:PA} to argue that $\forall\exists$ integer programs can be solved efficiently.

\begin{lemma}\label{thm:EisenbrandShmonin}
Suppose $A\in \matrices{\Z}{\xx}{\yy}$ and $B\in \matrices{\Z}{\yy}{\tv}$ are matrices and $\dd\in \Z^\tv$ is a vector. Then the satisfaction of the following sentence
$$\forall_{\myvv\in \Z^\yy} \left[(B\myvv\leq \dd) \implies \exists_{\uu\in \Z^\xx}\, (A\uu\leq \myvv)\right]$$
can be verified in time $f(\Delta,|\xx|,|\yy|,|\tv|)\cdot \left(\log \|\dd\|_\infty\right)^{\Oh(1)}$, where $\Delta=\max(\|A\|_\infty,\|B\|_\infty)$ and $f$ is a computable function.
\end{lemma}
\begin{proof}
    It suffices to combine \cref{thm:PA} with the observation that the considered sentence can be easily rewritten to an equivalent first-order sentence over $\langle \Zn,+,0,1,2,\ldots \rangle$, whose length depends only on $\Delta,|\xx|,|\yy|,|\tv|$ and whose constants are bounded by $\|\dd\|_\infty$. Indeed, every inequality present in the constraint $B\myvv\leq \dd$ can be rewritten to an equality of two terms as follows:
    \begin{itemize}[nosep]
    \item Every summand on the left hand side with a negative coefficient is moved to the right hand side. Also, if the constant on the right hand side is negative, it is moved to the left hand side. Thus, after this step, every side is a sum of variables multiplied by positive coefficients and positive constants.
    \item Every summand of the form $\alpha v$, where $\alpha$ is a positive coefficient and $v=\myvv[y]$ for some $y\in \yy$, is replaced by $\underbrace{v+v+\ldots+v}_\alpha$. Note here that $\alpha\leq \Delta$.
    \item Inequality is turned into an equality by existentially quantifying a nonnegative slack variable.
    \end{itemize}
    We apply a similar rewriting to the constraints present in $A\uu\leq \myvv$. After these steps, the sentence is a first-order sentence over $\langle \Zn,+,0,1,2,\ldots \rangle$.
\end{proof}

We remark that \cref{thm:EisenbrandShmonin} follows also from the work of Eisenbrand and Shmonin on $\forall\exists$ integer programs; see~\cite[Theorem~4.2]{DBLP:journals/mor/EisenbrandS08}. In~\cite{DBLP:journals/mor/EisenbrandS08} Eisenbrand and Shmonin only speak about polynomial-time solvability for fixed parameters, but a careful inspection of the proof shows that in fact, the algorithm works in fixed-parameter time with relevant parameters, as described in \cref{thm:EisenbrandShmonin}.


With \cref{thm:EisenbrandShmonin} in place, we present an algorithm to compute faithful decompositions. We first prove the following statement, which shows that a large fraction of a decomposition can be extracted~efficiently.

\begin{lemma}\label{lem:ComputingOneStepDecomp}
Given a matrix $D\in \matrices{\Z}{\yy}{\tv}$ and a non-zero vector $\bb\in \Z^{\tv}$, one can in time $f(\Delta,|\yy|,k)\cdot (\log \|\bb\|_\infty)^{\Oh(1)}$, where $f$ is a computable function, $\Delta=\|D\|_\infty$, and $k=|\tv|$, compute a faithful decomposition $\Bb_0$ of $\bb$ with respect to $D$ with the following property: the $\ell_\infty$-norms of all the elements of $\Bb_0$ are bounded by $g(\Delta,k)$ except for at most one element $\bb'$, for which it holds that
$$\|\bb'\|_1\leq \left(1-\frac{1}{2^{(k\Delta)^{\Oh(k)}}}\right)\cdot \|\bb\|_1.$$
\end{lemma}
\begin{proof}
By \cref{lem:ConformalDecomp}, $\bb$ admits a faithful decomposition $\Bb$ with respect to $D$ such that the order of $\Bb$ is at most $\Xi\coloneqq g(\Delta,k)$. Observe that there are $M\coloneqq (2\Xi+1)^k$ different vectors in $\Z^\tv$ of $\ell_\infty$-norm at most~$\Xi$, hence also at most $M$ different vectors present in $\Bb$. Therefore, there exists a vector $\bb_0$ appearing in $\Bb$ with multiplicity $\alpha\geq 1$ such that
$$\alpha\cdot \|\bb_0\|_1\geq \frac{1}{M}\cdot \|\bb\|_1.$$
Let $\alpha'$ be the largest power of two such that $\alpha'\leq \alpha$. Then we have
\begin{equation}\label{eq:nornica}\alpha'\cdot \|\bb_0\|_1\geq \frac{1}{2M}\cdot \|\bb\|_1.
\end{equation}

The algorithm guesses, by trying all possibilities, vector $\bb_0$ and power of two $\alpha'$. Note that there are at most $M\leq 2^{(k\Delta)^{\Oh(k)}}$ choices for $\bb_0$ and at most $\log \|\bb\|_1$ choices for $\alpha'$, as every faithful decomposition of $\bb$ consists of at most $\|\bb\|_1$ vectors. Having guessed, $\bb_0$ and $\alpha'$, we define decomposition $\Bb_0$ to consist~of
\begin{itemize}[nosep]
    \item $\alpha'$ copies of the vector $\bb_0$, and
    \item vector $\bb'\coloneqq \bb-\alpha'\bb_0$.
\end{itemize}
As $\Bb$ is a faithful decomposition, it follows that $\Bb_0$ defined in this way is a faithful decomposition as well, provided $\bb_0$ and $\alpha'$ were chosen correctly. Further,  we have $\|\bb_0\|\leq \Xi=g(\Delta,k)$ and, by~\cref{eq:nornica}, also
$$\|\bb'\|_1 \geq \left(1-\frac{1}{2M}\right)\cdot \|\bb\|_1.$$
Note that $2M=2\cdot (2\Xi+1)^k=2\cdot (2g(\Delta,k)+1)^k\in 2^{(k\Delta)^{\Oh(k)}}$, as promised.
So what remains to prove is that for a given choice of $\bb_0$ and $\alpha'$, one can efficiently verify whether $\Bb_0$ defined as above is indeed a faithful decomposition of $\bb$.

Before we proceed, observe that by \cref{lem:GraverDecomp}, every solution $\ww\in \Zn^{\yy}$ to $D\ww=\bb_0$ can be decomposed as
\begin{equation}\label{eq:rekin}
    \ww=\whww+\sum_{\gg\in \Gg} \gg,
\end{equation}
where $\whww\in \Zn^{\yy}$ is also a solution to $D\whww=\bb_0$ that additionally satisfies $\|\whww\|_\infty\leq (2k(\Xi+\Delta)+1)^{k}$, while $\Gg$ is a multiset of nonnegative vectors belonging to the Graver basis of $D$. Vector $\whww$ will be called the {\em{base solution}} for $\ww$.

With this observation, we can formulate the task of verifying faithfulness of $\Bb_0$ as a $\forall\exists$ integer program as follows. We need to verify whether for every $\myvv\in \Zn^\yy$ satisfying $D\myvv=\bb$, there is a vector $\myvv'\in \Zn^\yy$ satisfying $D\myvv'=\bb'$ and vectors $\ww_1,\ldots,\ww_{\alpha'}\in \Zn^\yy$ satisfying $D\ww_i=\bb_0$ for all $i\in \{1,\ldots,\alpha'\}$ such that $\myvv'+\sum_{i=1}^{\alpha'}\ww_i=\myvv$. Note that every vector $\ww_i$ will have a decomposition~\cref{eq:rekin}, with some base vector~$\whww_i$. Hence, to verify the existence of suitable $\myvv',\ww_1,\ldots,\ww_{\alpha'}$, it suffices to find the following objects:
\begin{itemize}[nosep]
 \item a vector $\myvv'\in \Zn^\yy$ satisfying $D\myvv'=\bb'$;
 \item for every $\whww\in \Zn^\yy$ such that $\|\whww\|_\infty\leq (2k(\Xi+\Delta)+1)^{k}$ and $D\whww=\bb_0$, a multiplicity $\gamma_{\whww}\in \Zn$ signifying how many times $\whww$ will appear as the base solution $\whww_i$; and
 \item for every Graver element $\gg\in \Graver(D)\cap \Zn^{\yy}$, a multiplicity $\delta_\gg\in \Zn$ with which $\gg$ will appear in the decompositions~\cref{eq:rekin} of the solutions $\mset{\ww_i\colon i\in \{1,\ldots,\alpha'\}}$ in total.
\end{itemize}
Note that by \cref{lem:GraverBound}, $\Graver(D)$ can be computed in time depending only on $\Delta$, $|\yy|$, and $k$.
The constraints that the above variables should obey, except for integrality and nonnegativity, are the following:
\begin{align*}
    D\myvv' & = \bb', \\
    \myvv' + \sum_{\whww} \gamma_{\whww}\cdot \whww + \sum_{\gg} \delta_\gg\cdot \gg & = \myvv,\\
    \sum_{\whww} \gamma_{\whww} & = \alpha' ,
\end{align*}
where summations over $\whww$ and $\gg$ go over all relevant $\whww$ and $\gg$ for which the corresponding variables are defined. It is easy to see that the feasibility of the integer program presented above is equivalent to the existence of suitable $\myvv',\ww_1,\ldots,\ww_{\alpha'}$; note here that the nonnegative elements of the Graver basis can be added to any solution of $D\ww=\bb_0$ without changing feasibility. As we need to verify the feasibility for every $\myvv\in \Zn^\yy$ satisfying $D\myvv=\bb$, we have effectively constructed a $\forall\exists$ integer program of the form considered in \cref{thm:EisenbrandShmonin}. It now remains to observe that the number of variables of this program as well as the absolute values of all the entries in its matrices are bounded by computable functions of $\Delta$, $|\yy|$, and~$k$. So we may use the algorithm of \cref{thm:EisenbrandShmonin} to solve this $\forall\exists$ program within the promised running time~bounds.
\end{proof}

We now give the full algorithm, which boils down to applying \cref{lem:ComputingOneStepDecomp} exhaustively.

\begin{lemma}\label{lem:ComputingCompleteDecomp}
Given a matrix $D\in \matrices{\Z}{\yy}{\tv}$ and a non-zero vector $\bb\in \Z^{\tv}$, one can in time $f(\Delta,|\yy|,k)\cdot (\log \|\bb\|_\infty)^{\Oh(1)}$, where $f$ is a computable function, $\Delta=\|D\|_\infty$, and $k=|\tv|$, compute a faithful decomposition $\Bb$ of $\bb$ with respect to $D$ of order at most $g(\Delta,k)$.
\end{lemma}
\begin{proof}
Again, start with a faithful decomposition $\Bb$ consisting only of $\bb$. Next, as long as $\Bb$ contains a vector $\bb'$ with $\|\bb'\|>g(\Delta,k)$, apply \cref{lem:ComputingOneStepDecomp} to $\bb'$, which yields a faithful decomposition $\Bb'$ of $\bb'$. Then replace $\bb$ with $\Bb'$ in $\Bb$ and continue. It is easy to see that $\Bb$ remains a faithful decomposition throughout this procedure, hence once all vectors in $\Bb$ have $\ell_\infty$-norm bounded by $g(\Delta,k)$, one can simply output $\Bb$.

It is easy to see that at each point of the procedure, $\Bb$ contains at most one vector with $\ell_\infty$-norm larger than $g(\Delta,k)$. Moreover, by \cref{lem:ComputingOneStepDecomp}, at every iteration the $\ell_1$-norm of this vector decreases by being multiplied by at most $1-\frac{1}{2^{(k\Delta)^{\Oh(k)}}}$. It follows that the total number of iterations performed by the procedure is at most $2^{(k\Delta)^{\Oh(k)}}\cdot \log \|\bb\|_1=2^{(k\Delta)^{\Oh(k)}}\cdot \log \|\bb\|_\infty$. Since by \cref{lem:ComputingOneStepDecomp} every iteration takes time $f(\Delta,|\yy|,k)\cdot (\log \|\bb\|_\infty)^{\Oh(1)}$, the claim follows.
\end{proof}

\subsection{Algorithm for \texorpdfstring{$n$}{n}-fold integer programming}

We are now ready to prove \cref{thm:nfold-main}.

\begin{proof}[Proof of Theorem~\ref{thm:nfold-main}]
Let $P=(C,\av,\mset{D_i,\bb_i,\cc_i\colon i\in I})$ be the given program. By adding some dummy variables and constraints if necessary, we may assume that all vectors $\bb_i$ are non-zero. The first step of the algorithm is to construct a new program $P'$, equivalent to $P$, in which all the right-hand sides $\bb_i$ will be bounded in the $\ell_\infty$-norm. This essentially boils down to applying \cref{lem:ComputingCompleteDecomp} to every $\bb_i$ as follows.

For every $i\in I$, apply \cref{lem:ComputingCompleteDecomp} to matrix $D_i$ and vector $\bb_i$, yielding a faithful decomposition $\Bb_i$ of $\bb_i$ with respect to $D_i$ of order at most $\Xi\coloneqq g(\Delta,|\tv|)$. Obtain a new uniform $n$-fold program $P'=(C,\av,\mset{D_{j},\bb_{j},\cc_{j}\colon j\in J})$ by replacing the single vector $\bb_i$ with the collection of vectors $\Bb_i$, for each~$i\in I$. Thus every index $j\in J$ originates in some index $i\in I$, and for all indices $j$ originating in $i$ we put $D_{j}=D_i$ and $\cc_{j}=\cc_i$.

Since collections $\Bb_i$ were faithful decompositions of vectors $\bb_i$ with respect to $D_i$, for all $i\in I$, it follows that the program $P'$ is equivalent to $P$ in terms of feasibility and the minimum attainable value of the optimization goal. This is because every solution to $P$ can be decomposed into a solution to $P'$ with the same optimization goal value using faithfulness of the decompositions $\Bb_i$, and also every solution to $P'$ can be naturally composed into a solution to $P$ with the same value.

We note that the program $P'$ is not computed by the algorithm explicitly, but in a high-multiplicity form. That is, decompositions $\Bb_i$ are output by the algorithm of \cref{lem:ComputingCompleteDecomp} by providing every vector present in $\Bb_i$ together with its multiplicity in $\Bb_i$. Consequently, we may describe $P'$ concisely by writing the multiplicity of every {\em{brick type}} present in $P'$, where the brick type of an index $j\in J$ is defined by
\begin{itemize}[nosep]
    \item the diagonal matrix $D_{j}$;
    \item the right-hand side $\bb_{j}$; and
    \item the optimization vector $\cc_{j}$.
\end{itemize}
In other words, indices $j,j'\in J$ with $D_j=D_{j'}$, $\bb_j=\bb_{j'}$, and $\cc_j=\cc_{j'}$ are considered to have the same brick type. Note that there are at most $(2\Delta+1)^{|\yy|\cdot |\tv|}$ different diagonal matrices $D_j$, at most $(2\Xi+1)^{|\tv|}$ different right-hand sides $\bb_j$, and at most $|I|\leq \|P\|$ different optimization vectors $\cc_j$, hence the total number of possible brick types is bounded $h(\Delta,|\yy|,|\tv|)\cdot \|P\|$, for some computable function $h$. For each such brick type, we store one positive integer of bitsize bounded by $\|P\|$ describing the multiplicity.

We now aim to solve the $n$-fold integer program $P'$ efficiently by formulating it as another integer program~$M$, and then relaxing $M$ to a mixed integer program $M'$ with a bounded number of integral variables. Intuitively, $M$ will determine multiplicities in which solutions to individual bricks are taken, and assign them to the blocks in an optimal manner. The solutions to bricks will not be guessed explicitly, but through their decompositions provided by \cref{lem:GraverDecomp}.

To formulate $M$, we need first some notation regarding $P'$. Let
\begin{itemize}[nosep]
    \item $\Diag=\matrices{\{-\Delta,\ldots,\Delta\}}{\yy}{\tv}$ be the set of all possible diagonal blocks; and
    \item $\RHS=\{-\Xi,\ldots,\Xi\}^\tv$ be the set of all possible right-hand sides.
\end{itemize}
Also, $I$ is the index set of the collection $\mset{\cc_i\colon i\in I}$ consisting of all possible optimization vectors.
Thus, the set of all possible brick types that may be present in $P'$ is $\Types\coloneqq \RHS\times \Diag\times I$, and we assume that $P'$ is stored by providing, for each type $(D,\bb,i)\in \Types$, the multiplicity $\cntTypes[D,\bb,i]$ with which the type $(D,\bb,i)$ appears in $P'$.

For $D\in \Diag$ and $\bb\in \RHS$, we let $\Base[D,\bb]$ be the set of all {\em{base solutions}} for $D$ and $\bb$: vectors $\whww\in \Zn^\yy$ such that $D\whww=\bb$ and $\|\whww\|\leq (2|\tv|(\Delta+\Xi)+1)^{|\tv|}$. Recall that by \cref{lem:GraverDecomp}, every nonnegative integer solution $\ww$ to $D\ww=\bb$ can be decomposed as
\begin{equation}\label{eq:ropucha}
\ww = \whww+\sum_{\gg\in \Gg} \gg,
\end{equation}
where $\whww\in \Base[D,\bb]$ and $\Gg$ is a multiset consisting of nonnegative vectors from the Graver basis. For convenience, denote $\Gra(D)\coloneqq \Graver(D)\cap \Zn^\tv$.

We may now define the variables of $M$. These will be:
\begin{itemize}[nosep]
    \item $\zeta^{D,\bb}_\whww$ for $D\in \Diag$, $\bb\in \RHS$, and $\whww\in \Base[D,\bb]$, signifying how many times in total the base solution $\whww$ will be taken for a brick with diagonal block $D$ and right-hand side $\bb$;
    \item $\delta^D_\gg$ for $D\in \Diag$ and $\gg\in \Gra(D)$, signifying how many times in total the Graver basis element $\gg$ will appear in the decompositions~\cref{eq:ropucha} of solutions to bricks; and
    \item $\omega^{D,\bb,i}_{\whww}$ for $D\in \Diag$, $\bb\in \RHS$, $i\in I$, and $\whww\in \Base[D,\bb]$, signifying how many times the base solution $\whww$ will be assigned to a brick of type $(D,\bb,i)$.
\end{itemize}
Next, we define the constraints of $M$:
\begin{align*}
&\ \ \ \sum_{D\in \Diag}\,\sum_{\bb\in \RHS}\left( \sum_{\whww\in \Base{[D,\bb]}} \zeta^{D,\bb}_{\whww}\cdot C\whww + \sum_{\gg\in \Gra(D)} \delta^D_\gg\cdot C\gg \right)  = \av;
& \tag{C1} \label{IP:FindSolAss1} \\
&\quad\ \ \sum_{i\in I} \omega^{D,\bb,i}_{\whww} = \zeta^{D,\bb}_{\whww} \qquad\qquad\qquad\qquad\  \textrm{for all }D\in \Diag, \bb\in \RHS, \textrm{ and } \whww\in \Base[D,\bb]; & \tag{C2} \label{IP:FindSolAss3} \\
& \sum_{\whww\in \Base[D,\bb]} \omega^{D,\bb,i}_{\whww} = \cntTypes[D,\bb,i] \qquad\quad \textrm{for all }D\in \Diag, \bb\in \RHS, \textrm{ and } i\in I; & \tag{C3} \label{IP:FindSolAss4} \\
&\quad\ \ \zeta^{D,\bb}_{\whww} \in \Zn \qquad \textrm{for all }D\in \Diag, \bb\in \RHS, \textrm{ and } \whww\in \Base[D,\bb]; &
\tag{C4} \label{IP:FindSolAss5}\\
& \quad\ \ \delta^D_\gg \in \Zn\qquad\ \ \   \textrm{for all }D\in \Diag \textrm{ and } \gg\in \Gra[D,\bb]; &
\tag{C5} \label{IP:FindSolAss6}\\
& \quad\ \ \omega^{D,\bb,i}_{\whww} \in \Zn \quad\ \,\textrm{for all }D\in \Diag, \bb\in \RHS, \whww\in \Base[D,\bb], \textrm{ and } i\in I. &
\tag{C6} \label{IP:FindSolAss7}
\end{align*}
Let us explain these constraints:
\begin{itemize}[nosep]
\item Constraint~\cref{IP:FindSolAss1} corresponds to the linking constraints of the $n$-fold program $P'$.
\item  Constraint~\cref{IP:FindSolAss3} assures that the numbers $\omega^{D,\bb,i}_{\whww}$ of base solutions $\whww\in \Base[D,\bb]$ assigned to bricks with different optimization vectors add up to the total number of such base solutions.
\item In a similar vein, Constraint~\cref{IP:FindSolAss4} states that
 the number of bricks of a particular brick type matches the total number of suitable base solutions assigned to them.
\item Finally, constraints~\cref{IP:FindSolAss5}, \cref{IP:FindSolAss6}, and \cref{IP:FindSolAss7} ensure the integrality and nonnegativity of our variables.
\end{itemize}
Last but not least, the objective function of $M$ is to
\begin{equation*}
\text{minimize}\quad \sum_{D\in \Diag}\ \sum_{\bb\in \RHS}\ \sum_{\whww\in \Base[D,\bb]}\ \sum_{i\in I}\ \omega^{D,\bb,i}_{\whww}\cdot \sca{\cc_i}{\whww} +\sum_{D\in \Diag}\ \sum_{\gg\in \Gra(D)}  \delta^D_\gg\cdot  \rev^{D}_\gg,
\end{equation*}
where for $D\in \Diag$ and $\gg\in \Gra(D)$, we define $\rev^D_\gg$ to be the minimum of $\sca{\cc_i}{\gg}$ over all $i\in I$ such that $D_i=D$. Here, the first summand is just the contribution of the base solutions to the optimization goal, while in the second summand we observe that every Graver basis element that appears in decompositions~\cref{eq:ropucha} can be assigned freely to any brick with the corresponding diagonal block $D$, hence we may greedily assign it to the brick where its contribution to the optimization goal is the smallest. Note that here we crucially use the assumption that the $n$-fold program $P'$ is uniform, because we exploit the fact that assigning both the base solutions and the Graver basis elements to different bricks has always the same effect on the linking constraints.

From the construction it is clear that the integer programs $P'$ and $M$ are equivalent in terms of feasibility and the minimum attainable value of the optimization goal. However, the number of variables of $M$ is too large to solve it directly: while the total number of variables $\zeta^{D,\bb}_{\whww}$ and $\delta^D_\gg$ is indeed bounded in terms of the parameters, this is not the case for the assignment variables $\omega^{D,\bb,i}_{\whww}$, of which there may be as many as $\Omega(|I|)$.

To solve this difficulty, we consider a mixed program $M'$ that differs from $M$ by replacing \cref{IP:FindSolAss7} with the following constraint:
\begin{align}
    & \quad\ \ \omega^{D,\bb,i}_{\whww} \in \myRn \quad\ \,\textrm{for all }D\in \Diag, \bb\in \RHS, \whww\in \Base[D,\bb], \textrm{ and } i\in I. &
\tag{C7'} \label{IP:FindSolAss7p}
\end{align}
That is, we relax variables $\omega^{D,\bb,i}_{\whww}$ to be fractional. We now observe that $M$ and $M'$ in fact have the same optima, because after fixing the integral variables, the program $M'$ encodes a weighted flow problem and therefore its constraint matrix is totally unimodular.

\begin{techclaim}\label{cl:rounding}
    The minimum attainable optimization goal values of programs $M$ and $M'$ are equal.
\end{techclaim}
\begin{subproof}
As $M'$ is a relaxation of $M$, it suffices to prove that $M'$ has an optimum solution that is integral.
Observe that after fixing the integral variables $\zeta^{D,\bb}_{\whww}$ and $\delta^D_\gg$, the only remaining constraints are \cref{IP:FindSolAss3} and~\cref{IP:FindSolAss4}. These constraints encode that values $\omega^{D,\bb,i}_{\whww}$ form a feasible flow in the following flow network:
\begin{itemize}[nosep]
    \item There is a source $s$, a sink $t$, and two set of vertices: $V_s\coloneqq \{(D,\bb,\whww)\colon D\in \Diag, \bb\in \RHS, \whww\in \Base[D,\bb]\}$ and $V_t\coloneqq \{(D,\bb,i)\colon D\in \Diag, \bb\in \RHS, i\in I\}$.
    \item For every $(D,\bb,\whww)\in V_s$, there is an arc from $s$ to $(D,\bb,\whww)$ with capacity $\zeta_{\whww}^{D,\bb}$.
    \item For every $(D,\bb,i)\in V_t$, there is an arc from $(D,\bb,i)$ to $t$ with capacity $\cntTypes[D,\bb,i]$.
    \item For every $D\in \Diag$, $\bb\in \RHS$, $\whww\in \Base[D,\bb]$, and $i\in I$, there is an arc from $(D,\bb,\whww)$ to $(D,\bb,i)$ with infinite capacity.
\end{itemize}
It is well known that constraint matrices of flow problems are totally unimodular, hence so is the matrix of program $M'$ after fixing any evaluation of the integral variables. Vertices of polyhedra defined by totally unimodular matrices are integral, hence after fixing the integral variables of $M'$, the fractional variables can be always made integral without increasing the optimization goal value. This shows that $M'$ always has an optimum solution that is integral.
\end{subproof}

Consequently, to find the optimum value for the initial $n$-fold program $P$, it suffices to solve the mixed program $M'$. As $M'$ can be constructed in time $f(\Delta,|\yy|,|\tv|)\cdot \|P\|^{\Oh(1)}$ and has $f(\Delta,|\yy|,|\tv|)$ integral variables, for some computable function $f$, we may apply any fixed-parameter algorithm for solving mixed programs, for instance that of Lenstra~\cite{Lenstra1983}, to solve $M'$ within the promised running time bounds.

Finally, let us remark that while the presented procedure outputs only the optimum value for $P$, one can actually find an optimum solution to $P$ by (i) finding an optimum mixed solution to $M'$, (ii) finding an optimum integer solution to $M'$ by solving the fractional part integrally using e.g. the flow formulation constructed in \cref{cl:rounding}, and (iii) translating the obtained optimum solution to $M'$ back to an optimum solution to $P$. We leave the details to the reader.
\end{proof}

\printbibliography

\newpage
\appendix
\section{Proof of \texorpdfstring{\Cref{obs:4block}}{Observation \ref{obs:4block}}}\label{app:4block}

In this section, we prove \Cref{obs:4block}, restated below for convenience.

\fourblock*

\begin{proof}
    Let $P$ be the input program:
    \begin{align}
 \xx\in \Zn^k,\ \yy_t\in \Zn^k, &\nonumber\\
     B\xx+\sum_{t=1}^n C_t\yy_t = \av, & \qquad \textrm{and}\label{eq:linking} \\
 A_t\xx + D_t\yy_t =\bb_t &\qquad \textrm{for all }t=1,2,\ldots,n,\label{eq:local}
    \end{align}
    where $A,B,C,D_i$ are integer $k\times k$ matrices and $\av,\bb_i$ are integer vectors of length $k$. Denote $A=[a_{ij}]_{i,j\in [k]}$, $B=[b_{ij}]_{i,j\in [k]}$, and $C=[c_{ij}]_{i,j\in [k]}$, and recall that we assume that all entries in $A,B,C$ have absolute values bounded by $n$. Also, let $\av=(a_1,\ldots,a_k)^\intercal$, $\xx=(x_1,\ldots,x_k)^\intercal$, and $\yy_t=(y_{t,1},\ldots,y_{t,k})^\intercal$, for $t=1,\ldots,n$. Variables $\xx$ will be called {\em{global variables}} and variables $\yy_t$ will be called {\em{local variables}}. Similarly, constraints~\eqref{eq:linking} will be called {\em{linking}}, and constraints~\eqref{eq:local} will be called {\em{local}}. Local variables and local constraints come in $t$ {\em{groups}}, each consisting of $k$ variables/constraints.

    We will gradually modify the program $P$ so that at the end, in the modified blocks $A,B,C$ all the coefficients will belong to $\{-1,0,1\}$. During the modification, we do not insist on all blocks being always square matrices; this can be always fixed at the end by adding some zero rows or zero columns.

    First, we reduce large entries in $C$. Note that the $i$th linking constraint takes the form
    $$\sum_{j=1}^k b_{ij}\,x_j + \sum_{t=1}^n \sum_{j=1}^k c_{ij}\,y_{t,j}= a_i,$$
    or equivalently,
    $$\sum_{j=1}^k b_{ij}\,x_j + \sum_{j=1}^k c_{ij}\sum_{t=1}^n  y_{t,j}= a_i.$$
    Hence, for every pair $(i,j)\in [k]\times [k]$ we introduce a new global variable $z_{ij}$ together with the constraint
    $$-z_{ij}+\sum_{t=1}^n y_{t,j}=0,$$
    and in the $i$th linking constraint, we replace the summand $c_{ij}\sum_{t=1}^n  y_{t,j}$ with the summand $c_{ij}z_{ij}$. It is easy to see that after this modification, only $\{0,1\}$ entries will appear in the blocks $C$. This comes at the price of increasing the number of global variables and the number of linking constraints by $k^2$, and possibly adding some large (but bounded by $n$ in absolute value) entries to the block $B$.

    Next, we remove large entries in $A$. Suppose then $a_{ij}>1$ for some $i,j\in [k]$; recall that we assume also that $a_{ij}\leq n$. In program $P$, entry $a_{ij}$ occurs in summands of the form $a_{ij}x_j$; every group of local constraints contains one such summand. Therefore, the idea is to (i) copy the variable $x_j$ to $a_i$ new local variables, (ii) use one new linking constraint to define a new global variable $z'_{ij}$ equal to the sum of the copies, and hence equal to $a_{ij}x_j$, and (iii) replace each usage of the summand $a_{ij}x_j$ in local constraints with just the variable $z'_{ij}$. This idea can be executed as follows:
    \begin{itemize}[nosep]
        \item[(i)] For every $t=1,\ldots,n$, introduce two new local variables $p_{ij,t},p'_{ij,t}$ together with two new local constraints:
        \begin{align*}
         & -x_j+p'_{ij,t}=0;\qquad \textrm{and}\\
         & p_{ij,t}-p'_{ij,t}=0\textrm{\, when \,}t\leq a_{ij}\qquad \textrm{and}\qquad p_{ij,t}=0 \textrm{\, when \,} t>a_{ij}.
        \end{align*}
        \item[(ii)] Add a new global variable $z'_{ij}$ together with a new linking constraint $$-z'_{ij}+\sum_{t=1}^n p_{ij,t}=0.$$
        \item[(iii)] Replace every usage of summand $a_{ij}x_{ij}$ with variable $z'_{ij}$ in every local constraint within~$P$.
    \end{itemize}

    Entries $a_{ij}\in \{-n,\ldots,-2\}$ can be removed similarly: we perform the above construction for $|a_{ij}|$, and then replace each summand $a_{ij}x_{ij}$ with $-z'_{ij}$ instead of $z'_{ij}$. Entries with large absolute values in $B$ can be removed in exactly the same way.

    It is straightforward to see that by applying all the modifications presented above, we obtain a uniform $4$-block program with $\Oh(k^2)$ global variables, $\Oh(k^2)$ linking constraints, $\Oh(k^2)$ local variables in each group of local variables, and $\Oh(k^2)$ local constraints in each group of local constraints. Moreover, all entries in the blocks $A,B,C$ belong to $\{-1,0,1\}$, while blocks $D_i$ got modified only by adding some $\{-1,0,1\}$ entries. Hence, we can just solve the modified program using the assumed fixed-parameter algorithm for feasibility of uniform $4$-block programs parameterized by the dimensions of the block and the maximum absolute value of any entry appearing in the constraint matrix.
\end{proof}

\section{Hardness of two-stage stochastic integer programming with large diagonal entries}
\label{sec:nphard}

We conclude our discussion of two-stage stochastic integer programming by a short argument showing that it is indeed necessary to include $\Delta=\max_{i\in I} \|D_i\|_\infty$ among the parameters. This is because when $\Delta$ is unbounded the problem becomes strongly $\mathsf{NP}$-hard already for blocks of size $k=16$, as we prove below. This rules out even a running time of the form $f(k) \cdot (\|P\|+ \Delta)^{\Oh(1)}$, and shows that the dependence on $\Delta$ needs to be superpolynomial.

We give a reduction from the 3-SAT problem. Let $N$ denote the number of \emph{variables} and $M$ the number of \emph{clauses} in the input formula. The reduction creates an instance of \stageFeas{} with only one \emph{global} variable $x \in \Zn$, whose role is to encode a satisfying 3-SAT assignment. Let $p_1, p_2, \ldots, p_N$ be the first $N$ primes. In our encoding, $p_i \mid x$ if and only if the $i$-th SAT variable is set to false. The IP instance contains $M$ diagonal blocks, each corresponding to a single clause.

We use a certain gadget to construct the blocks. The gadget consists of five ILP variables, named $a,b,c,d,e\in \Zn$, and three constraints that involve those five variables as well as the global variable $x$. Let $p$ be a fixed prime (not an ILP variable). The key properties of the gadget are:
\begin{itemize}[nosep]
\item if $p \mid x$, then every feasible assignment to the five local variables has $b=0$;
\item if $p \nmid x$,  then every feasible assignment to the five local variables has $b=1$; and
\item for every $x \in \Zn$ there exists a feasible assignment to the five local variables.
\end{itemize}
We leave it to the reader to verify that the following gadget satisfies the desired properties.
\begin{align*}
-x + p \cdot a + b + c & = 0 \\
(p-2) \cdot b - c - d &= 0 \qquad \text{(equivalent to $c \leq (p-2) \cdot b$)} \\
b + e &= 1 \qquad \text{(equivalent to $b \leq 1$)}
\end{align*}

Each block consists of three such independent gadgets, instantiated for the three primes corresponding to the three SAT variables involved in the corresponding clause. Moreover, each block contains one additional constraint (and one additional local variable) ensuring that the three $b$-variables (or their negations in case of negative literals) sum up to at least one (minus the number of negative literals in the clause), which corresponds to the clause being satisfied.

Summarizing, the resulting two-stage stochastic ILP instance is:
\[\setcounter{MaxMatrixCols}{20}
D_i = \begin{bmatrix}
p_{v_{i,1}} & 1 & 1 & & &   &&&&&   &&&&& \\
& p_{v_{i,1}}\!\!-\!2 & -1 & -1 & & &&&&& &&&&& \\
& 1 & & & 1 & &&&&& \\
&&&&& p_{v_{i,2}} & 1 & 1 & & & &&&&& \\
&&&&& & p_{v_{i,2}}\!\!-\!2 & -1 & -1 & & &&&&& \\
&&&&& & 1 & & & 1 & &&&&& \\
&&&&& &&&&& p_{v_{i,3}} & 1 & 1 & & & \\
&&&&& &&&&& & p_{v_{i,3}}\!\!-\!2 & -1 & -1 & & \\
&&&&& &&&&& & 1 & & & 1 & \\
&t_{i,1}&&&& &t_{i,2}&&&& &t_{i,3}&&&& -1 \\
\end{bmatrix},
\]
\[
A_i = \begin{bmatrix} -1 \\ 0 \\ 0 \\ -1 \\ 0 \\ 0 \\ -1 \\ 0 \\ 0 \\ 0 \end{bmatrix}, \quad
\bb_i = \begin{bmatrix} 0 \\ 0 \\ 1 \\ 0 \\ 0 \\ 1 \\ 0 \\ 0 \\ 1 \\ 1 - \ell_i\end{bmatrix}, \quad
\text{for } i \in [M],
\]
where, for every for $i \in [M]$ and $j \in \{1,2,3\}$, $v_{i,j}$ denotes the index of the variable in the $j$-th literal in the $i$-th clause, $t_{i,j}=1$ if this is a positive literal and $t_{i,j}=-1$ otherwise, and $\ell_i$ equals the number of negative literals in the $i$-th clause, i.e., $\ell_i = (3-(t_{i,1}+t_{i,2}+t_{i,3}))/2$.

From the preceding discussion it follows that the obtained integer program is feasible if and only if the input 3-SAT formula is satisfiable. Each $A_i$ is a $10 \times 1$ matrix, each $D_i$ is a $10 \times 16$ matrix, and we have $\max_{i\in I} \|D_i\|_\infty \leq p_N = \Oh(N \log N)$ by the Prime Number Theorem. We conclude that the two-stage stochastic integer programming feasibility problem is strongly $\mathsf{NP}$-hard already for $k=16$.

\end{document}